\documentclass[10pt]{article}

\usepackage[colorlinks]{hyperref}
\usepackage{amsmath}
\usepackage{amsthm}                 
\usepackage{latexsym,url}
\newcommand{\be}{\begin{equation}}
\newcommand{\ee}{\end{equation}}
\newcommand{\pa}{\partial}
\newcommand{\al}{\alpha}
\newcommand{\De}{\Delta}
\newcommand{\de}{\delta}
\newcommand{\f}{\phi}

\newcommand{\La}{\Lambda}
\newcommand{\la}{\lambda}

\newcommand{\si}{\sigma}

\newcommand{\cB}{{\cal B}}
\newcommand{\cC}{{\cal C}}
\newcommand{\cZ}{{\cal Z}}
\newcommand{\cO}{{\cal O}}
\newcommand{\cS}{{\cal S}}
\newcommand{\cE}{{\cal E}}
\newcommand{\cR}{{\cal R}}
\newcommand{\cF}{{\cal F}}
\newcommand{\cL}{{\cal L}}
\newcommand{\bR}{{\bf R}}
\newcommand{\bN}{{\bf N}}
\newcommand{\bn}{{\bf n}}
\newcommand{\bZ}{{\bf Z}}

\newcommand{\supp}{{\rm supp }}

\newcommand{\norm}[1]{\left|\left|#1\right|\right|}  
\newcommand{\abs}[1]{\left|#1\right|}                
\newtheorem{thm}{Theorem}
\newtheorem{lem}[thm]{Lemma}
\newtheorem{remark}[thm]{Remark}
\numberwithin{equation}{section}   
\allowdisplaybreaks

\begin{document}
\title{Dipole$-$Dipole Correlations for the sine-Gordon Model}
\author{ Guowei Zhao\thanks{Email: \url{gzhao2@buffalo.edu}. Email after Sept 1, 2011: \url{zgwei@ustc.edu}}\\ Dept. of Mathematics \\SUNY at Buffalo \\Buffalo, N.Y. 14260}
\maketitle
\begin{abstract}
We consider the dipole-dipole correlations for the two-dimensional Coulomb gas/sine-Gordon model for $\beta>8\pi$ by a renormalization group method. First we re-establish the renormalization group analysis for the partition function using finite range decomposition of the covariance. Then we extend the analysis to the correlation functions. Finally, we show a power-law decay characteristic of the dipole gas.
\end{abstract}
\section{Introduction}
The two dimensional Coulomb gas is a particle system in $\bR^2$ with electric charges  $\pm 1$ interacting with Coulomb potential. Its equilibrium state is defined by a Gibbs measure on configuration space. 
The equilibrium state can be identified with the vacuum for a sine-Gordon quantum field theory.

Since the work of Kosterlitz and Thouless \cite{KT73}, which predicts the phase transition point, 
several rigorous works  have appeared on Coulomb gas/sine-Gordon model. Remarkably, Fr\"ohlich and Spencer treated Coulomb lattice gas at large $\beta$ by an inductive method in \cite{FS81}, and then Dimock and Hurd gave a renormalization group  analysis on both  $\beta>8\pi$ (infrared) and $\beta<8\pi$ (ultraviolet) problems with small activity in \cite{DH00}. Recently, Falco studied the pressure along a small piece of Kosterlitz-Thouless line near $\beta=8\pi$.

In the present paper, we study the dipole-dipole correlation for the two-dimensional Coulomb gas/sine-Gordon model at low temperature by a renormalization group method. The proof is mainly composed by two parts: the first part is to re-establish the uniform bounds for the partition function of the continuum model via finite range decomposition of the covariance, the second part is to extend the  analysis to the correlation functions, and  prove a power-law decay for long-distance dipoles.

Now we start with the brief description of Coulomb gas/sine-Gordon model. The standard canonical partition function of the two dimensional Coulomb gas with inverse temperature $\beta$ and activity $z/2$ is given on the torus $\La_M=\bR^2/L^M\bZ^2$
\be
Z(\La_M) = \sum_{n=0}^{\infty}\frac{z^n}{2^nn!}\left[\sum_{q_1,\dots,q_n}\int_{\La_M\times\cdots\times\La_M} \exp \left(-\frac{\beta}{2}\sum_{i,j}q_iq_jv(x_i-x_j)\right)dx_1\dots dx_n\right]   
\ee
Here the summation is over unit charges $q_i=\pm 1$, $i=1,\dots,n$. The interaction between particles is given by the potential
 \be v(x-y)=\frac{1}{2\pi}\log\abs{x-y}\ee
And after sine-Gordon transformation, it can be formally rewritten as  the cutoff expression
\be \label{opdef}
Z=\int \exp\left(z\int_{\La_M}\ \cos(\phi(x))dx\right)\ d\mu_{\beta
v}(\phi). \ee
The integral $d\mu$ is with respect to Gaussian measure with mean 0 and covariance $\beta v$. For a precise definition, we approximate $v$ by 
\be  \label{vdef}
 v_{M}( x-y)=\sum_{j=0}^{M-1} C\left(\frac{x-y}{L^j}\right)
\ee
Here the covariance $C(x-y)=0$ if $\abs{x-y}>L$, i.e., $C$ is finite range. This was used previously in \cite{BMS03}  for $d\geq3$.
We will show that if $\int\rho=0$, then as $M\rightarrow\infty$, the term $\int_y v_M(x-y)\rho(y)dy$ will converge to $\int_y v(x-y)\rho(y)dy+\int_y w(x-y)\rho(y)dy$, where $w(x-y)$ 
vanishes when $\abs{x-y}>1$. We tolerate this short distance modification since we are mainly interested in the long distance/infrared problem.

First we study the infrared problem of  the partition function of sine-Gordon following Brydges-Yau renormalization group method \cite{BY90}. 

And as in \cite{D09}, we consider
\be \label{pdef}
Z(\La_M,z,\sigma)=\int \exp \left( zW(\La_M,\phi)-\sigma V(\La_M,\phi) \right) d\mu_{\beta v_M}(\phi)
\ee
where
\be \label{v1def}
W(\La_M,\phi)=\int_{\La_M}\cos(\phi(x))dx
\ee
\be \label{v2def}
V(\La_M,\phi)=\int_{\La_M} (\partial\phi)^2dx
\ee
Here the inclusion of the term $\sigma\int (\partial\phi)^2$ corresponds to a modification of the vacuum dielectric constant. We regard $\sigma$ as an adjustable parameter. In a sense we will be tuning $\sigma$ so the actual dielectric constant is unity.

Renormalization group involves repeated integration with respect to Gaussian measure $d\mu_{\beta C}$ in each scale. To use the similar R.G. treatments as in \cite{DH00} and \cite{BMS03}, we introduce a local structure for the sequence of densities. A closed  polymer $X$ means a connected union of some   unit closed blocks in $\La$. A polymer activity $K(X,\phi)$ is a function which depends on  $\phi$ only in polymer $X$. And for given polymer activities $V$ and $K$, the exponential operation is defined as
\be
\cE xp(\Box e^{-V}+K)(X,\phi)=\sum_{X_i} e^{-V(X  \setminus \cup_i X_i,\phi)}\prod_iK(X_i,\phi)
\ee 
where $\Box$ is the characteristic function of  unit blocks $\De$, and the summation is over collections of disjoint polymers $\{X_i\}$ in $X$.  Refer to \cite{B09} for the interpretation of the  notation $\cE xp$.

Then we can state the result for the partition function as:
\begin{thm}   
\label{pthm}   
Let $\beta > 8 \pi$, 
let  $L$ be sufficiently large, 
and  let $\abs{z}$ be  sufficiently small.
Then there exists $\sigma_0$ so that the  function $Z$ defined by (\ref{pdef}) can be written for $0\leq j\leq M$ as:
\be
Z= e^{\cE_j}\int \cE xp(\Box e^{-V_j} +K_j)(\La_{M-j},\phi)  d\mu_{\beta v_{M-j}}(\phi)
\ee
Here
\begin{itemize}
\item $\de\cE_j=\cE_j-\cE_{j-1}$ is bounded in the volume $\La_{M-j}.$ 
\item The background potentials $V_j$ satisfy:
\be
V_j(\Delta)=\sigma_j\int_{\Delta} (\partial\phi)^2
\ee
and $V_j(X)=\sum_{\Delta\subset X} V_j(\Delta)$, where  $\sigma_j$ are  constants.

\item The polymer activities $K_j$ are uniformly bounded in volume, and with certain specified norm
\be \norm{K_{j}}_{j}\leq \delta^j\epsilon\ee
with $\delta=\cO(1) \max\{L^{-2},L^{2-\beta/4\pi}\}$.  And $\cO(1)$ means a constant independent of $j$ and $L$.
\item 
With  chosen $\sigma_0$, the flow $(\sigma_j, \norm{K_j})\rightarrow (0,0)$ as $j\rightarrow \infty$.
\end{itemize}
\end{thm}

This result is essentially the same as Theorem 1 in \cite{DH00}, but here we are using the finite range decomposition of the covariance \cite{BGM04} in the continuum case \cite{BMS03} instead of the regular splitting covariance in \cite{DH00}.

Then we consider the correlation functions for the sine-Gordon model. For a discussion, see \cite{FS81}. The Coulomb gas is expected to behave like a dipole gas for large $\beta$, i.e. low temperature. Since a charge at $x$ is represented by a field $\phi(x)$, a dipole as $x$ and $x+\epsilon \textbf{n}$ is represented by
 \[\frac{\phi(x+\epsilon \textbf{n})-\phi(x)}{\epsilon}\approx \textbf{n}\cdot \partial\phi(x)\]
 Thus dipole-dipole correlations are studied with the field $\partial\phi$.

Let \be(\partial\phi,\rho)=\la_1\bn_1\cdot\partial\phi(a)+\la_2\bn_2\cdot\partial\phi(b)\ee for $a,b\in\La_M$, with $\abs{a-b}\gg1$, $\bn_1$ and $\bn_2$ are two unit vectors. Then $\rho$ is defined implicitly and $\int \rho=0$.

Now we consider the generating functional:
\be  \label{gfdef}   
Z(\rho)=<e^{i(\partial\phi,\rho)}>=Z^{-1}\int e^{i(\partial\phi,\rho)} e^{zW} d\mu_{\beta v_M}(\phi)
\ee
then the truncated field correlation function will be
\be   \label{tcdef}   
\left<\partial\phi(a)\partial\phi(b)\right>^T=(-1)\cdot \left[\frac{\partial^2}{\partial \la_1\partial\la_2}\log Z(\rho) \right] \Big|_{\rho=0}
\ee
here the decay of (\ref{tcdef}) is the interesting problem.

First, by the renormalization group analysis as Theorem \ref{pthm} for the partition function, for $Z(\rho)$ we also have:
\begin{thm} \label{05182011}
For sufficiently small $\abs{z}$ 
\be
Z(\rho)=e^{\cE_j(\rho)}\int \cE xp(\Box e^{-V_j(\phi)} +K_j(\phi,\rho))(\La_{M-j})  d\mu_{\beta v_{M-j}}(\phi)
\ee
hold for any $0\leq j\leq M-1$, where $V_j$'s are $\rho-$independent, and  polymer activities satisfy:
\be
\norm{K_{j}(\rho)}_{j}\leq \de^{j}\epsilon'
\ee
where $\epsilon'$ is a $\rho-$dependent constant. 
And $K_j(\rho)$ has pinning property for any $j$.
\end{thm}

The idea of the proof is to track the growth of those new polymer activities $K_j(X,\phi,\rho)$ as in Theorem \ref{pthm}. 
The new $\rho-$dependent or $\la-$dependent polymer activities  still keep the localization property and $2\pi-$ translation invariant property. Moreover the polymer activities $K_j(X,\phi,\rho)$ have \textit{pinning property}, i.e. they are the same as the vacuum polymer activities $K_j(X,\phi,0)$ except a few terms for which $X$ contains one or both dipoles. 

Furthermore, based on the pinning property, we are able to prove the main decay result for the truncated correlation function:
\begin{thm} \label{cthem} 
Let $\beta > 8 \pi$,  
let  $L$ be sufficiently large, 
and  let $\abs{z}$ be  sufficiently small.
Then for $\epsilon_0>0$, there is a constant $C$ which is independent of $M$ such that
\be
\abs{\left<\partial\phi(a)\partial\phi(b)\right>^T}\leq C\abs{a-b}^{-2(1-\epsilon_0)}
\ee
\end{thm}

This supports the general picture that Coulomb gas behaves like dipole gas over low temperature. Earlier attempts on the field correlation function on this model can be found in \cite{H95} \cite{MK91}.

\section{Renormalization Group}
\subsection{Finite Range Covariance}
\subsubsection{Decomposition Construction} \label{FRC}

Let $g$ be any non-negative translation-invariant  $C^{\infty}$ function with compact support in $\{x\in\bR^2:\abs{x}\leq 1/2\}$. Then if $u=g*g$, then $u$ will be a positive definite  $C^{\infty}$ function with compact support in $\{x:\abs{x}\leq 1\}$.

We define:

\be
C(x-y)=\int_{1}^{L} \frac{1}{l} u\left(\frac{x-y}{l}\right)dl
\ee
for any integer $L\geq 2$. 

Then $C$ is finite range up to $L$, i.e.,
\be C(x-y)=0\quad\mbox{if}\quad \abs{x-y}\geq L\ee
and 
\be
C(0)=\int_{1}^{L} \frac{1}{l} u(0)dl=u(0)\log L=\cO(1)\log L
\ee
Note that this construction is true for any $g$. Later we will choose $g$ such that $u(0)=1/2\pi$.

If we continue to define the multi-scale decomposition for the corresponding scale:
\be
v_{M-j}(x-y)=\sum_{k=0}^{M-j-1}C(L^{-k}(x-y))
\ee
then we have
\be 
v_{M-j}(x-y)=C(x-y)+v_{M-j-1}\left(\frac{x-y}{L}\right)
\ee
Under this notation, we have the finite range decomposition property over different scales: $v_{M-j}(x)=0$ if $\abs{x}\geq L^{j+1}$ for all $j=0,1,2,\dots.$

This multi-scale construction is analogous to \cite{BMS03}\cite{BGM04}. Here our construction is on the continuum torus case. \cite{BGM04} discussed the finite range decomposition on lattice case for three or higher dimensional space. A more comprehensive discussion about the  finite range decomposition of the two dimensional massive/massless covariance can be found in \cite{F10}.

\subsubsection{Consistency}
Now we consider the case in $\bR^2$ temporarily.
We now check that as $M\rightarrow\infty$,
the covariance $v_M$ has the same long distance behavior as  the two dimensional Coulomb potential $v(x-y)=\frac{1}{2\pi}\log \abs{x-y}$ which is the fundamental solution of the Laplacian $-\De$. 
We also need to check that $v_M$ allows the Coulomb gas/sine-Gordon identification.

\begin{lem} \label{dimock0311}
\mbox{}

\begin{enumerate}
	\item Let \be
	v_{0,M}(r)=\int_1^{L^M}u(\frac{r}{l})\frac{dl}{l}\ee
	Then for any $r>0$,
	\be\lim_{M\rightarrow\infty} (v_{0,M}(r)-v_{0,M}(1))=const\cdot\log r+w(r)\ee
	where $w(r)$ vanishes for $r\geq1$.
	\item For $dist (x,\supp \rho)>0$,
\be
\begin{split}
&\lim_{M\rightarrow \infty} \int_{y\in\La_M} v_M(x-y)\rho(y)dy\\
= &\left\{ \begin{array}{ll}
\int_{y\in\La_M}\left( v(x-y)+w(x-y)\right)\rho(y)dy & \textrm{if  $\int \rho=0$}\\
\infty & \textrm{if   $\int \rho\neq 0$}
\end{array} \right.
\end{split}
\ee
where  $w(x-y)=0$ if $\abs{x-y}\geq1$.
\end{enumerate}
 
\end{lem}
\begin{proof}
\mbox{}

\begin{enumerate}
\item 
For $r>0$, we write  
\be \label{04192011b}
v_{0,M}(r)=\int_1^{L^M}u(\frac{r}{l})\frac{dl}{l}=\int_0^{L^M}u(\frac{r}{l})\frac{dl}{l}-\int_0^{1}u(\frac{r}{l})\frac{dl}{l}
\ee
Denote the two above integrals as 
\be \label{04192011}
\tilde{ v}_{0,M}(r)=\int_0^{L^M}u(\frac{r}{l})\frac{dl}{l}
\qquad \mbox{ and }\qquad w_0(r)=\int_0^{1}u(\frac{r}{l})\frac{dl}{l}
\ee

For the term $w_0(r)$, note that $\frac{r}{l}>r\geq 1$ for $l<1$, so $u(\frac{r}{l})=0$ because of the support of $u$, which implies $w_0(r)$ vanishes when $r\geq 1$.

Now for the term $\tilde{ v}_{0,M}(r)$, its derivative
\be
\tilde{v}'_{0,M}(r)=\int_0^{L^M}u'(\frac{r}{l})\frac{dl}{l^2}
\ee
is convergent as $M\rightarrow \infty$, i.e, 
\be
\lim_{M\rightarrow \infty} \tilde{v}'_{0,M}(r)=\int_0^{\infty} u'(\frac{r}{l})\frac{dl}{l^2}=\frac{1}{r}\int_0^{\infty}u'(\frac{r}{l_1})\frac{dl_1}{l_1^2}=const\cdot \frac{1}{r}
\ee
And similarly $\abs{\tilde{v}'_{0,M}(r)}\leq const\cdot 1/r$ for any $M$. Then  by Dominated Convergence Theorem,
\be
\lim_{M\rightarrow \infty} (\tilde{v}_{0,M}(r)-\tilde{v}_{0,M}(1))=\int_1^r \lim_{M\rightarrow \infty} \tilde{v}'_{0,M}(s)ds=const\cdot\log r
\ee
The result follows from combining this with $w_0(r)$.
\item First note that $v_M(x-y)$ is a function of $\abs{x-y}$, i.e, $v_M(x-y)=v_{0,M}(\abs{x-y})$. Similarly $\tilde{v}_{M}(x-y)=\tilde{v}_{0,M}(\abs{x-y})$ and $w(x-y)=w_0(\abs{x-y})$. So for $dist (x,\supp \rho)\geq 0$, if $\int \rho=0$, then for the term $\tilde{v_M}$, we have
\begin{align} \label{04022011a}
\lim_{M\rightarrow \infty}\int \tilde{v}_M(x-y)\rho(y)dy 
&=\lim_{M\rightarrow \infty}\int (\tilde{v}_M(x-y)-\tilde{v}_M(1))\rho(y)dy \\
&=\int const\cdot\log(x-y)\rho(y)dy\\
&=\int v(x-y)\rho(y)dy
\end{align}
Here $v(x-y)=const\cdot\log (x-y)$ is the fundamental solution of the two dimensional Laplacian $-\De$. 

Then for the term $w$, note that  $w(r)$ is always  independent of $M$ whenever $M\rightarrow \infty$ or not. Therefore the limit $ \lim_{M\rightarrow \infty}\int w(x-y)\rho(y)dy$ exists. 

The claimed limit follows from combining this with \ref{04022011a}. 

If $\int\rho\neq 0$, then by (\ref{04192011b}):
\be
\begin{split}
&\int_{y\in\La_M}v_M(x-y)\rho(y)dy\\
=&\int_{y\in\La_M}(v_M(x-y)-v_M(1))\rho(y)dy+\int_{y\in\La_M}v_{0,M}(1)\rho(y)dy
\end{split}
\ee
As $M\rightarrow \infty$, $\int_{y\in\La_M}(v_M(x-y)-v_M(1))dy$ converges. 
Now by $u(0)=1/2\pi$,  \be \begin{split}
&\lim_{M\rightarrow \infty}\int_1^{L^M}u(\frac{1}{l})\frac{dl}{l}\\
=&\lim_{M\rightarrow \infty}\Big[\int_1^{L^M}(u(\frac{1}{l})-\frac{1}{2\pi})\frac{dl}{l}+\int_1^{L^M}\frac{1}{2\pi}\frac{dl}{l}\Big]\\
\geq &\lim_{M\rightarrow \infty}\Big[const\int_1^{L^M}\frac{dl}{l^2}+\frac{M}{2\pi}\log L \Big]\\
=&+\infty
\end{split}\ee
where $const=\min\limits_{0\leq t\leq 1/l} u'(t)$ is finite for $l\leq 1$ since $u$ is smooth. 
Hence if $\int\rho(y)dy>0$, we have
\be
\begin{split}
\lim_{M\rightarrow \infty}\int_{y\in\La_M}v_M(1)\rho(y)dy&=\lim_{M\rightarrow \infty}v_{0,M}(1)\int_{y\in\La_M}\rho(y)dy\\
&=\lim_{M\rightarrow \infty}\int_1^{L^M}u(\frac{1}{l})\frac{dl}{l}\int_{y\in\La_M}\rho(y)dy\\
&\geq\infty
\end{split}
\ee
which implies the divergence\be\lim_{M\rightarrow \infty}\int_{y\in\La_M}v_M(x-y)\rho(y)dy =\infty\ee
Similarly when $\int \rho<0$.
\end{enumerate}
\end{proof}
Here the point is the sequence $v_M$ itself does not converge to $v$. The approach only holds for $\rho$ under neutrality condition of $\rho$. Also the behavior of $w$ is not important since we are investigating the long distance behavior of the system. In other words, we are taking a short distance(ultraviolet) cutoff of the covariance.

To see the Coulomb gas/sine-Gordon connection, we start with the  partition function:
\begin{align}
Z(\La_M)=&\int \exp\left(z\int_{\La_M}\ \cos(\f(x))dx\right)\ d\mu_{
v{M}}(\phi) \\
=& \sum_{n=0}^{\infty} \frac{z^n}{n!}\sum_{\vec{e}} \int \exp\left(\sum_{i,j=1}^n -e_ie_jv_M(x_i-x_j) \right)dx_1\cdots dx_n\\
=& \sum_{n=0}^{\infty} \frac{z^n}{n!}\sum_{\vec{e}} \int \exp\left(-(\rho,v_M\rho) \right)dx_1\cdots dx_n
\end{align}
where 
$
\rho(x)=\sum_i e_i\de_{x_i}
$ and the summation $\sum\limits_{\vec{e} } $ denotes the summation over all unit charges $(\vec{x},\vec{e})=(x_1,e_1;\dots;x_n,e_n)$ where each $x_i$ is associated with a charge $e_i$, $1\leq i\leq n$. 
and the inner product is defined as
\be
(\rho,C\rho)=\sum_{x,y}\rho(x)C(x-y)\rho(y)=\iint \rho(x)C(x-y)\rho(y) dxdy
\ee
 for any covariance $C$.

The above argument is the rigorous statement, now we consider to take its formal volume limit  as $M\rightarrow \infty$. According to Lemma \ref{dimock0311},
\begin{align} \label{05192011a}
&\lim_{M\rightarrow \infty} \int \exp\left(z\int_{\La_M}\ \cos(\f(x))dx\right)\ d\mu_{
v{M}}(\phi)\\
=& \lim_{M\rightarrow \infty} \sum_{n=0}^{\infty} \frac{z^n}{n!}\sum_{\vec{e} \atop \mbox{neutral}} \int \exp\left(-(\rho,v_M\rho) \right)dx_1\cdots dx_n
\end{align}
Here the $\sum\limits_{\vec{e} \atop \mbox{neutral}} $ is the summation over all unit charges such that their summation is $0$, which is neutrality condition. The contributions from the non-zero summations (charged sectors) is $0$ since $(\rho,v_M\rho)\rightarrow \infty$ by Lemma \ref{dimock0311}.
In other words, eventually we have
\be
\begin{split}
&\lim_{M\rightarrow \infty}\int \exp\left(z\int_{\La_M}\ \cos(\f(x))dx\right)\ d\mu_{
v{M}}(\phi) \\ =&  \sum_{n=0}^{\infty} \frac{z^n}{n!}\sum_{\vec{e}\atop \mbox{neutral}} \exp\left(-\frac{1}{2}\sum_{i,j}e_ie_jv(x_i-x_j)\right) \label{05192011b}
\end{split}
\ee

And the  expression in \ref{05192011a} and \ref{05192011b} is just the formal limit in $\bR^2$. 

Now for the for the torus case $\La_M$,   we take as an approximate covariance
\be  v_M(x-y)=v_{0,M}(\abs{x-y}) \ee
where $\abs{x-y}$ indicates the metric on torus $\La_M$,
With this covariance, our model formally has the correct infinite volume limit as $m\rightarrow \infty$.

\subsection{Norm} \label{NormSection}
To control the growth of the polymer activities during the RG transform, let us specify the norms first. In \cite{DH00}, the polymer activities are smooth functions on fields $\phi$ only, but here we are still using the similar norm for the $\rho-$dependent polymer activities. Similar treatments can be found in \cite{DH92} and \cite{BK94}.

For fields $\phi$ defined on (closed) polymer $X$, we define a real separable Hilbert space $H(X)$ as following. For unit block $\Delta$, we consider the Sobolev space $W_s(\mathring{\De})$ on its interior $\mathring{\De}$ with the norm
\be
\norm{\phi}_{W^s(\mathring{\De})}=\left(\sum_{\abs{\alpha}\leq s} \int_{\mathring{\De}}\abs{\partial^{\alpha}\phi(x)}dx\right)^{1/2}
\ee
We also consider the real Banach space $C^r(\De)$  of $C^r$ functions $\phi:\De\rightarrow\bf{R} $ which are also $r-$th continuously differentiable on all the closed unit block $\De$, with the norm defined as
\be
\norm{\phi}_{C^r(\De)}=\sup_{x\in\De}\max_{\abs{\alpha}\leq r}\abs{\partial^{\alpha}\phi(x)}
\ee
Then by the Sobolev embedding theorem, there is a constant $C$ independent of the choice of unit block $\De$ such that
\be
\norm{\phi}_{C^r(\De)} \leq C\norm{\phi}_{W^s(\mathring{\De})}
\ee
with $s>r+1$.

Now let $\tilde{H}(X)$  be the finite direct sum of the Hilbert spaces $W^s(\mathring{\De})$ for $\De\subset X$. Then we define $H(X)$ to be the subspace of $\tilde{H}(X)$ induced by the following condition: $\phi\in H(X)$ if for any neighboring unit blocks $\De_1,\De_2\subset X$, the  $C^r$ images by the Sobolev embedding of $\phi_{\De_1}$ and   $\phi_{\De_2}$ match as well as the derivatives on the common boundary component $\De_1\cup\De_2$. By the embedding theorem, $H(X)$ is a real Hilbert space with the norm
\be
\norm{\phi}_{H(X)}=\left(     \sum_{\De\subset X}    \sum_{\abs{\alpha}\leq s} \int_{\mathring{\De}}\abs{\partial^{\alpha}\phi(x)}dx        \right)^{1/2}
\ee 
Also let $C^r(X)$ be the real Banach space as above with the norm
\be
\norm{\phi}_{C^r(X)}=\sup_{x\in X}\max_{\abs{\alpha}\leq r}\abs{\partial^{\alpha}\phi(x)}
\ee
Then for $s>r+1$, we have an embedding $H(X)\hookrightarrow C^r(X)$ and
\be
\norm{\phi}_{C^r(X)} \leq C\norm{\phi}_{H(X)}
\ee
where the constant $C$ is independent of $X$.

The construction here is essentially the same as the norm in \cite{DH00}, but with a sightly modification by \cite{A07} for the completeness of the normed space of polymer activities. Also it preserves all the estimates in \cite{BDH98} \cite{DH00}.

We assume that $K(X,\phi,\rho)$ is a smooth function of $\phi\in H(X)$ and analytic within a small ball in $\rho$. Here $K(X,\phi)=K(X,\phi,\rho=0)$ for vacuum case.
\begin{enumerate}

\item For $n=0,1,2,\dots$, the $n-$th derivative with respect to $\phi$ is a multi-linear functional on $f_i\in H(X)$ and evaluated as:
\be
\begin{split}
&K_n(X,\phi,\rho;f_1,\dots,f_n)\\ =& \frac{\de^n}{\de t_1\dots\de t_n}   K(X,\phi+t_1f_1+\dots+t_nf_n,\rho) \Big| _{t=0}
\end{split}
\ee
Then we define 
\be
\label{varnorm}
\norm{K_n(X,\phi,\rho)} =  \sup_{f_i \in H(X) \atop \|f_i\|_{C^r(X)} \leq 1} |K_n(X,\phi,\rho;f_1,...,f_n) | 
\ee
Note that this norm is stronger than the norm on $H(X)$.

\item Next, for a further parameter $h>0$ to be specified, we define:
\be
\norm{K(X,\phi,\rho)}_{h}=\sum_{n=0}^{\infty} \frac{h^n}{n!}\norm{K_n(X,\phi,\rho)}
\ee

\item For large field regulator, which is a functional of $\phi$ of the form:
\be
 G( \kappa,X, \phi ) =
 G'(\kappa, X, \phi ) \de G(\kappa, \pa X, \phi )   \label{gee}
\ee
where
\begin{align}   G'(\kappa, X, \phi )&=\exp\left( \kappa\norm{\phi}^2_{H(X)}   \right)\\
 \de G(\kappa, \pa X, \phi )  
&=   \exp  \left( \kappa c  \sum_{ | \al |=1}\int_{\pa X}  |\pa^{\al} \phi |^2  \right)
\end{align}
with constants $\kappa, c \leq 1$  to be specified.
we define:
\be
\norm{K_n(X,\rho)}_{h,G}=\sup_{\phi\in H(X)}\left(\norm{K_n(X,\phi,\rho)}_h G(X,\phi)^{-1}\right)
\ee

\item Finally,  for a large set regulator $\Gamma(X)$ of the form
$
\Gamma(X)=A^{|X|}
$ 
 where $A\geq 1$ , we define:
\be
\norm{K(\rho)}_{h,G,\Gamma}=\sum_{X\supset\Delta}\Gamma(X)\norm{K(X,\rho)}_{h,G}
\ee
the summation is over all polymers containing unit block $\Delta$.

Also, we define the corresponding norm associated with the regulator  $\Gamma_p(X)=2^{p\abs{X}}\Gamma(X)$ for any $p=\pm 1,\pm 2,\dots.$
\end{enumerate}

Note that for the norms defined above, we have the following estimate:
\be \label{normproperty}
\norm{K(X)K(Y)}_{h,G(X)G(Y)} \leq \norm{K(X)}_{h,G(X)} \cdot \norm{K(Y)}_{h,G(Y)} 
\ee
for any polymers $X$ and $Y$.
The proof of this property can be found in \cite{BDH98}.

The important thing here is that the space of all the smooth polymer activities  on $X$ with the norm defined above is a Banach space. 

\subsection{Renormalization Group Transformation} \label{fundmentaltransform}
In this section, we will set up the R.G. flow $(V_j,K_j)\rightarrow(V_{j+1},K_{j+1})$, where
\begin{align}
&\int \cE xp\left(\Box e^{-V_j}+K_j\right)\left(\La_{M-j},\phi_L+\zeta\right)d\mu_C(\zeta)\\
=&\exp\left(\sum_{X\subset\La_{M-j-1}} F_j(X)\right)\cdot \cE xp\left(\Box e^{-V_{j+1}}+K_{j+1}\right)\left(\La_{M-j-1},\phi\right)
\end{align}
This comes in three steps: fluctuation, extraction and scaling steps as in \cite{DH00}. Here we are using finite range decomposition of the covariance $C$, therefore it's convenient to follow the treatments in \cite{BMS03}. Without confusion, we write $V=V_j$, $\La=\La_{M-j}$ and $K=K_j$ as the beginning.

Now we introduce some terminologies in R.G. analysis first. A block is an unit block in $\La_{M-j}$, and a polymer, usually denoted as $X$, is the union of some blocks. The size of $X$ is the number of blocks in $X$, denoted as $\abs{X}$. The polymer $X$ is small if $X$ is connected and $\abs{X}\leq 2^d$ ($d=2$ in this present paper ). We say $X$ is large if it is not small.

\subsubsection{Fluctuation}

We start with the expression of the integrand:
\be \label{Expdef}
\cE xp\left(\Box e^{-V}+K\right)=\sum_{\{X_i\}}e^{-V(\La \setminus \cup_{i} X_i)} \prod_{i} K(X_i)
\ee
where the summation is over all the collections of disjoint polymers $X_i$ in $X$. 

Write $X=\cup X_i$, $\La \setminus \cup_{i} X_i=X^c$ and the closure of $X^c= \overline{X^c}$. Then one can write
\be
e^{-V(X^c,\phi)} =\prod\limits_{\Delta\subset \overline{X^c}} e^{-V(\Delta,\phi)}
\ee
i.e., $e^{-V}$ has factorization property.

Define the polymer activity $P$ as
\be
P(\Delta,\zeta,\phi)=e^{-V(\Delta,\phi+\zeta)}-e^{-V(\Delta,\phi)}
\ee
Note that here $V$ depends on $\phi$ only, where $\zeta$ is the integration variable in the fluctuation step.
Now
\be
e^{-V(X^c,\phi+\zeta)}=\prod_{\Delta\subset\overline{X^c}} \left( e^{-V(\Delta,\phi)}+P(\Delta,\phi,\zeta)\right)
\ee
By expanding the product and substituting into (\ref{Expdef}), we have:
\be
\begin{split}
&\cE xp\left(\Box e^{-V}+K\right)(\La)\\
=&\sum_{\{X_i\},\{\Delta_j\}} e^{-V(\La \setminus(\cup X_i)\cup(\cup \Delta_j))} \prod_{i} K(X_i) \prod_j P(\Delta_j)
\end{split}
\ee

Define the $L-$polymers $Y$ to be:
\be
Y=\overline{(\cup X_i)\cup(\cup \Delta_j)}^L
\ee
and $Y_1,\dots,Y_P$ be the connected subsets of $Y$. Note that here we group $X_i$ and $\De_j$ into connected components of $Y$. Then 
\be \label{ExpRdef}
\cE xp\left(\Box e^{-V}+K\right)(\phi+\zeta,\La)=\cE xp_L \left(\Box e^{-V}+\cB K\right)(\phi,\zeta,\La)
\ee
where
\be \label{reblockformula}
\cB K(Y,\phi,\zeta) =\sum_{\{X_i\},\{\Delta_j\}\rightarrow Y}  e^{-V(X_0)} \prod_{i} K(X_i,\phi+\zeta) \prod_j P(\Delta_j,\phi,\zeta)
\ee
where \[X_0=Y\setminus (\cup X_i)\cup(\cup \De_j)\] and the summation is defined as
\be
\sum_{\{X_i\},\{\Delta_j\}\rightarrow \{Y\}} =\sum_{M,N}\frac{1}{N!M!}\sum_{(X_1,\cdots,X_N)\atop (\De_1,\cdots,\De_M)}
\ee over the maps from $\{X_i\},\{\Delta_j\}$ into $\{Y_1,\dots,Y_p\}$ . As in \cite{BMS03}, the expression (\ref{ExpRdef}) is a sum over the products of polymer activities $\cB K$ where the closed disjoint polymers are separated by a distance larger than $L$. This will give us the advantage to utilize  the finite range property of the covariance $C$ in the fluctuation step.

Finally, after applying fluctuation integral, the fluctuation integral factorizes by the range of the covariance  $C$ as \textbf{Fluctuation Formula}:
\be \label{flucformula}
\begin{split}
&\int \cE xp\left(\Box e^{-V}+K\right)\left(\La,\phi+\zeta\right)d\mu_C(\zeta) \\
=& \cE xp_L\left(\Box_L e^{-V}+ \cB K^{\#}  \right)\left(\La,\phi\right)
\end{split}
\ee
where
\begin{itemize}
\item $\cE xp_L$ indicates that the expansion is over $L-$ polymers.
\item $\Box_L$ is the characteristic function on ``$L-$unit'' blocks.
\item The integration $\#$ in the expression of $\cB K$ denotes the integration with the measure $d \mu_C(\zeta)$.
\end{itemize}

\subsubsection{Scaling}
To continue the R.G. analysis on the next scale, we define the new polymers such that
\be \label{scaleformula}
\cE xp_L(\Box_L e^{-V}+K)(L\La,\phi_L)=\cE xp(\Box e^{-V_{L^{-1}}}+\cS K)(\La,\phi)
\ee
where the rescaled field $\phi_L(x)=\phi(x/L)$. 

It is straightforward that  one can choose:
\be
V_{L^{-1}}(\Delta,\phi)=V(L\Delta,\phi_L)
\ee
and
\be \label{scaleformula00}
(\cS K)(L^{-1}X,\phi)=K(X,\phi_L)
\ee
for $L$ polymers $X$.

The formula (\ref{scaleformula}) is the \textbf{Rescaling Formula}. For us, it turns out that $V(L\Delta,\phi_L)=V(\De,\phi)$.
Now if we combine the  fluctuation and rescaling steps now, we will have:
\be
\int \cE xp(\Box e^{-V}+K)(\La,\phi_L+\zeta) d\mu_C(\zeta)=\cE xp(\Box e^{-V_{L^{-1}}}+(\cS \cB K)^{\natural})(L^{-1}\La,\phi)
\ee
Here $\natural$ denotes the integration with the measure $d\mu_{C_{L^{-1}}}(\zeta)$, where the scaled covariance is defined as $C_{L^{-1}}(x-y)=C(L(x-y))$. And later for simpler notations, we also denote $\cF K=(\cS\cB K)^{\natural}$.

Then for ordinary $1-$polymer $Z$,
\be \label{SBformula}
\cS\cB K(Z,\phi,\zeta)=\sum_{\{X_i\},\{\Delta_j\}\rightarrow LZ }  e^{-V(X_0,\phi_L)} \prod_{i} K(X_i,\phi_L+\zeta_L) \prod_j P(\Delta_j,\phi_L,\zeta_L)
\ee
where \[X_0=LZ\setminus  (\cup X_i)\cup(\cup \De_j)\]

\subsubsection{Extraction} \label{extraction}

To isolate the fast-growth  terms of $K$, we need extraction formula, i.e., seeking for some $F$ 
\be \label{extraformula0}
\cE xp(\Box e^{-V}+K)(\La)=\cE xp(\Box e^{-V(F)}+\cE(K,F))(\La)
\ee

We assume $F$ has the form
\be
F(X,\phi)=\sum_{\Delta\subset X} F(X,\Delta,\phi)
\ee
where $F(X,\Delta,\phi)$ depends on the $\phi$ in $\De$ only, i.e., we are assuming that $F$ has localization property. 
And for any fixed unit block $\Delta$,
\be
V_F(\Delta)= \sum_{X\supset \Delta} F(X,\Delta)
\ee
We will specify the $F(X,\Delta)$ with some special forms later. For polymer $X$, $V_F(X)=\sum_{\De\subset X} V_F(\De)$.

Therefore in (\ref{extraformula0}), if we denote $V(F)$ as $V'$, then we can choose
\be
V'_F(\De)=V(\De)-V_F(\De)
\ee
Then the corresponding $\cE(K,F)$ is determined by this chosen $V'$. More precisely, the formula of $\cE(K,F)$ is given by:
\be \label{extraformula00}
\begin{split}
&\cE(K,F)(W)\\
=&\sum_{\{Z_j\}\{Y_k\}\rightarrow W} e^{-V'(W\setminus Y)}\prod_j\left( e^{-F(Z_j,Z_j\cap \overline{Y^c})}-1\right)\prod_k\tilde{K}(Y_k)
\end{split}
\ee
where $Y=\cup_k Y_k$, and the summation is over the collections of disjoint polymers {$Z_j$} and collections of distinct {$Y_k$} such that  each $Z_j$ intersects both $Y$ and  $Y^c$, $\{Z_j\}$ and $\{Y_k\}$ are overlap connected and $W=(\cup Z_j)\cup(\cup Y_k)$, and
\be
F(Z,Z\cap \overline{Y^c})=\sum_{\Delta\subset Z\cap \overline{Y^c}} F(Z,\Delta)
\ee
\be
\tilde{K}(X)=K(X)-e^{-V(X)}\left(e^F-1\right)^+(X)
\ee
where the $\cdot^+$ operation is defined as
\be
J^+(X)=\sum_{\{X_i\}\rightarrow X}\prod_i J(X_i)
\ee
where the summation is over all the distinct polymers {$X_i$} such that $\cup X_i=X$.

And the linear part of $\cE(K,F)$ in $K$ and $F$ has the form:
\be
\cE_1(K,F)=K-Fe^{-V}
\ee

Furthermore, to continue the R.G. transformations, we need to keep both $V_j$ and $K_j$ to be sufficiently small for any $j$. Therefore we need to isolate some terms from the circle product $\cE xp$, i.e., seeking  
\be \label{extraformula}
\cE xp(\Box e^{-V}+K)(\La)=e^{\sum_{X\subset \La} F_0(X)}\cE xp(\Box e^{-V''}+\cE(K,F_0,F_1))(\La)
\ee

Here $F=F_0+F_1$ and we assume  both $F_0$ and $F_1$ has localization property. And now with chosen $F_0$ and $F_1$, 
we define
\be
W(X)=\prod_{\De\subset X}e^{\sum_{Y\supset\De}F_0(Y,\De)}
\ee
and $W(0)=1$. Then with $F=F_0+F_1$, the formula (\ref{extraformula0}) becomes:
\begin{align}
\cE xp(\Box e^{-V}+K)(\La)&=\cE xp(\Box e^{-V(F)}+\cE(K,F))(\La)\\
&=W(\La)\cE xp(\Box W^{-1}e^{-V(F)}+W^{-1}\cE(K,F))(\La)
\end{align}
Hence we can define
\be
 e^{-V''}(X)=(W^{-1}e^{-V(F)})(X)
\ee
and
\be
\cE(K,F_0,F_1)(X)=(W^{-1}\cE(K,F_0+F_1))(X)
\ee
Here the direct computation shows that $e^{-V''}$ is $F_0$ independent and
\begin{align} \label{05262011}
e^{-V''}(X)=\left(W^{-1}e^{-V(F)}\right)(X)
= e^{-V(F_1)}(X)
\end{align}
Then (\ref{extraformula}) is the desired \textbf{Extraction Formula}.
And the original treatments of such extraction made here can be found in \cite{BDH98}\cite{BMS03} etc. 

\subsection{General Estimates}
We have constructed the convergent formulas during the first three sections. More precisely, if we  combine the fluctuation \ref{flucformula}, rescaling \ref{scaleformula} and extraction \ref{extraformula} steps together, the R.G transform will be:
\begin{align}
&\int \cE xp\left(\Box e^{-V_j}+K_j\right)\left(\La_{M-j},\phi_L+\zeta\right)d\mu_C(\zeta)\\
=&\exp\left(\sum_{X\subset\La_{M-j-1}} F_{0,j}(X)\right)\cdot \cE xp\left(\Box e^{-V_{j+1}}+K_{j+1}\right)\left(\La_{M-j-1},\phi\right)
\end{align}
where \be K_{j+1}=\cE((\cS\cB K_j)^{\natural},F_{j,0},F_{j,1}) \ee

Now we also need to track the growth of polymer activities under those transformations. First let us specify some assumptions on the renormalization group process. And we will make some reasonable assumptions as the necessary  conditions of the renormalization group transformations:
\begin{enumerate}
	\item $\norm{K}_{h,G,\Gamma}$ is sufficiently small;
	\item The constants in the regulator satisfy 
	$\kappa c^{-1} L^{2s-2}$ is sufficiently small, $\sigma/\kappa$ and $\sigma/\de\kappa$ are sufficiently small; and $\cO(1)\leq \kappa h^2\leq \cO(1)$; and $2\de \kappa\leq\kappa$;
	\item The polymer activity $V(\De,\phi)$ satisfies
	\be
	\norm{e^{-V(\De)}}_{h,G(\de\kappa)}\leq 2 \label{06062011}
	\ee
	\item For complex $z(X)$, the extraction $F$ satisfies:
	\begin{align} \label{stabilitycondition}
	&\sup_{\abs{z(X)}\de f(X)\leq 1} \norm{ \exp \left(-V(\Delta)-\sum_{X\supset \De} z(X)F(X,\De)\right)}_{h,G} \leq 4
	\end{align}
	where $\de f(X)$ is a
	constant such that $\norm{\de f}_{\Gamma_{p}}$ is sufficiently small. This condition is analogous to the stability condition in \cite{BDH98}, which is also similar to the condition in \cite{DH00} without background potential $V$.
\end{enumerate}

To control the norms of the polymer activities in intermediate steps, we introduce some new regulators.
Let us define the regulators $G_{L^{-1}}$ as:
\begin{align}
&G_{L^{-1}}(\kappa,X,\phi)\\=& G(\kappa, LX, \phi_{L}) \\
=&\exp\left( \sum_{\De\subset X}\kappa\sum_{1\leq\abs{\alpha}\leq s} L^{-2\abs{\alpha}+2} \int_{\mathring{\De}}\abs{\partial^{\alpha}\phi}^2 +\kappa c\sum_{\abs{\alpha}=1} L\int_{\partial X}     \abs{\partial^{\alpha}\phi}^2               \right)
\end{align}
And  define the  intermediate large field regulator
\be
\hat{G}_{}(\kappa,\delta\kappa,X,\phi,\zeta)=G_{}(\kappa,X,\phi+\zeta)G_{}(\delta\kappa,X,\phi)G_{}(\delta\kappa,X,\zeta)
\ee
The norm associated with this regulator is for the $\zeta$ dependent polymer activities before the fluctuation step.

\begin{lem} \label{rges}For any $p,q>0$, under the  conditions as stated above, we have 
\begin{enumerate}
\item Bound on Extraction:
\begin{align}
\norm{\cE (K,F)}_{h,G,\Gamma_p}&\leq\cO(1)\Big( \norm{K}_{h,G,\Gamma_{p+2}}+\norm{\de f}_{\Gamma_{p+4}}\Big)\\
\norm{\cE (K,F_0,F_1)}_{h,G,\Gamma_p}&\leq\cO(1) \left(\norm{K}_{h,G,\Gamma_{p+3}}+\norm{\de f}_{\Gamma_{p+5}}\right)
\end{align}
\item Bound on Rescaling and Reblocking:
\be
\norm{\cS \cB K}_{h,\hat{G}_{L^{-1}},\Gamma_p}\leq \cO(1)L^2\Big(\norm{K}_{h,G,\Gamma_{p-q+3}}+\abs{\sigma/\delta\kappa}\Big)
\ee
\item Bound on Gaussian integration:
\begin{align}
\norm{K^{\natural}}_{h,G(\kappa+\delta\kappa),\Gamma_{p}} &\leq \norm{K}_{h,\hat{G}_{L^{-1}}(\kappa),\Gamma_{p+1}}\\
\norm{K^{\#} }_{h,G_L(\kappa+\delta\kappa),\Gamma_{p}} &\leq \norm{K}_{h,\hat{G}(\kappa),\Gamma_{p+1}}
\end{align}
\end{enumerate}
\end{lem}
\begin{proof}
We prove the results one by one.
\begin{enumerate}
\item 
This bound is similar to Theorem 5 and Theorem 6 in \cite{BDH98}
or
Theorem 8 in \cite{DH00}.
First we write the extraction as
\begin{align}
&\cE(K,F)(W)\\=&\sum_{\{Z_j\}\{Y_k\}\rightarrow W} e^{-V'(W\setminus Y)}\prod_j\left( e^{-F(Z_j,Z_j\cap \overline{Y_c})}-1\right)\prod_k\tilde{K}(Y_k)\\
=&\sum_{\{Z_j\}\{Y_k\}\rightarrow W}  e^{-V'(W\setminus Y)}
        \prod_k\tilde{K}(Y_k)\\
       &\qquad\qquad \prod_j \frac{1}{2 \pi i}\int \frac{dz_{j}}{z_{j}(z_{j}-1)}
        \exp (- z_{j} F(Z_{j},Z_j\cap \overline{Y_c}))
\end{align}
 The integral is over the circles $|z_j|\de f(Z_j)= 1$, where $\de f$ is a constant depending on $Z_j$.

Then by the multiplicative property of the norms, we  have
\begin{align}
        &\norm{\cE(K,F)(W)}_{h,G}\\
\leq&\sum_{\{Z_j\}\{Y_k\}\rightarrow W} 
        \prod_k \norm{ \tilde K(X_k)}_{h,G}
        \prod_j \cO(1) \de f (Z_j)\\
    &\cdot\sup_{|z_j| \de f(Z_j) \leq 1} 
        \norm{ \exp( -V'(W\setminus Y)- \sum _{j} z_{j} F(Z_j, Z_j\cap \overline{Y_c}))}_{h,G}
\end{align}
By the  stability assumption \ref{stabilitycondition}, the last factor is bounded by
\begin{align}
 \prod_{\De \subset W\setminus Y}
        \norm{ \exp( -V'(\De)- \sum _{j} z_{j} F(Z_j,\De))}_{h,G}\leq 4^{\abs{W\setminus Y}} \leq 4^{\abs{Z}}
\end{align}
Now we substitute the original definition of the summations
\[\sum_{\{X_i\},\{Y_j\}}=\sum_{N,M}\frac{1}{N!M!}
\sum_{(X_1,\dots,X_N), (Y_1,\dots,Y_M)}\]
where the sum is over ordered sets, but otherwise the restrictions
apply.
We multiply by $\Gamma_p(W)$, identify $4^{|W|}\Gamma_p(W) = \Gamma_{p+2}(W)$  and use
$ \Gamma_{p+2}(W) \leq \prod_i\Gamma_{p+2}(X_i)
\prod_j\Gamma_{p+2}(Z_j) $
which follows  from the overlap connectedness.
Then sum over $W$ with a   pin, and use a spanning tree
argument in \cite{BY90} and the small norm hypotheses to obtain
\be
\begin{split} 
        &\norm{\cE(K,F)}_{h,G,\Gamma_p}\\
\leq & \sum_{N \geq 1, M \geq 0} \frac{(N+M)!}{N!M!}
        (\cO(1))^{N+M}
        \norm{\tilde K}^N_{h,G,\Gamma_{p+2}}
        \norm{\de f}_{\Gamma_{p+4}} ^M \\
\leq & \sum_{N \geq 1, M \geq 0} (\cO(1))^{N+M}2^{N+M}
\norm{\tilde K}^N_{h,G,\Gamma_{p+2}}
        \norm{\de f}_{\Gamma_{p+4}} ^M \\
\leq & \frac{\cO(1)\norm{\tilde K}_{h,G,\Gamma_{p+2}}}{1-\cO(1)\norm{\tilde K}_{h,G,\Gamma_{p+2}}}\sum_{ M \geq 0} (\cO(1))^{M}2^{M}
\norm{\de f}_{\Gamma_{p+4}} ^M \\
\leq & \cO (1) \norm{\tilde K}_{h,G,\Gamma_{p+2} } \label{04032011a}
 \end{split} \ee
The last steps are by the combinatoric inequality  $(N+M)!/N!M!  \leq 2^{N+M}$ and the summation formula for geometric series and under the assumption that $\norm{K}$ and $\de f$ are sufficiently small.

Now note that  $\tilde K=K- (e^{F}-1)^+$, we write
\be
        (e^{F}-1)^+(Y)
= \sum_{\{Y_j\}}\prod_j \frac{1}{2 \pi i}
        \int \frac{dz_{j}}{z_{j}(z_{j}-1)}
        \exp(  z_{j} F(Y_{j}))
\ee
 now with the integral over  $|z_j|\de f(Y_j) = 1$.
Then by a similar argument as above and using the pre-assumed 
condition \ref{stabilitycondition},  we have
\be  \| (e^{F} -1 )^+(Y) \| _{h,G(\kappa)} \leq
 2^{|Y|} \  \sum_{\{Y_j\}}\prod_j \cO(1) \de f(Y_j)
\ee
and hence
\be  \| (e^{F} -1 )^+ \| _{h,G(\kappa),   \Gamma_{p+2}} \leq
\sum_{N=1}^{\infty}(\cO(1))^N \|\de f\|^N_{\Gamma_{p+4}}
\leq \cO(1)  \|\de f\|_{\Gamma_{p+4}}   \label{04032011b} \ee
Eventually the final estimate follows from combining \ref{04032011a} and \ref{04032011b} together.

The estimate on $\norm{\cE (K,F_0,F_1)}_{h,G,\Gamma_p}$ is similar and by the property that $W^{-1}(X)\leq 2^{\abs{X}}$.

\item 
Recall the definition  (\ref{SBformula}), 
since $X_0,\cup X_i$ and $\cup \De_j$ are disjoint, after taking the functional derivatives 
we  have:
\be
\begin{split}
	&\norm{\cS\cB K(Z,\phi,\zeta)}_h\\
\leq &\sum_{\{X_i\},\{\Delta_j\}\rightarrow \{LZ \}} \norm{ e^{-V(X_0)}(\phi_L)}_h \\
&\qquad\qquad \prod_{i}\norm{K(X_i,\phi_L+\zeta_L)}_h
\prod_j \norm{P(\Delta_j,\phi_L,\zeta_L)}_h
\end{split} \ee

Now since both sides involves $\zeta$ fields, we shall multiply by the corresponding regulator $\hat{G}$ to control it. Then 
\be
\begin{split}
	&\norm{\cS\cB K(Z)}_{h,\hat{G}_{L^{-1}}}\\
=&\sup_{\phi,\zeta}\left(\norm{\cS\cB K(Z,\phi,\zeta)}\hat{G}_{}(\kappa,\delta\kappa,LZ,\phi_L,\zeta_L)^{-1}\right)\\
\leq &\sup_{\phi,\zeta}\sum_{\{X_i\},\{\Delta_j\}\rightarrow \{LZ \}} \norm{ e^{-V(X_0)}(\phi_L)}_h \prod_{i} \norm{K(X_i,\phi_L+\zeta_L)}_h \\&\qquad\qquad\prod_j \norm{P(\Delta_j,\phi_L,\zeta_L)}_h\hat{G}_{}(\kappa,\delta\kappa,LZ,\phi_L,\zeta_L)^{-1}\\
\leq & \sup_{\phi,\zeta}\sum_{\{X_i\},\{\Delta_j\}\rightarrow \{LZ \}} \norm{ e^{-V(X_0,\phi_L)}}_h G^{-1} (\de\kappa,X_0,\phi_L)\\
 &\prod_{i} \norm{K(X_i,\phi_L+\zeta_L)}_hG^{-1}(\kappa,X_i,\phi_L+\zeta_L)\\
\prod_j& \norm{P(\Delta_j,\phi_L,\zeta_L)}_h G^{-1}_{}(\delta\kappa,\De_j,\zeta_L)G_{}^{-1}(\kappa,\De_j,\phi_L+\zeta_L)G^{-1}_{}(\delta\kappa,\De_j,\phi_L)
\end{split}
\ee
For the first term, note that $V(L\De,\phi_L)=V(\De,\phi)$ and $V$ is defined on any scale, 
and $Z\setminus (L^{-1}(\cup_i X_i\cup\cup_j \De_j)=L^{-1}X\subset Z$.
Then by the assumption \ref{06062011},  after multiplying the large field regulator, it is bounded by
\be \begin{split}
&\sup_{\phi} \norm{ e^{-V(X_0,\phi_L)}}_h G^{-1}(\de\kappa,X_0,\phi_L)\\
=&\sup_{\phi} \norm{ e^{-V(L^{-1}X_0,\phi)}}_h G^{-1}_{L^{-1}}(\de\kappa,L^{-1}X_0,\phi)\\
\leq& \norm{ e^{-V(L^{-1}X_0)}}_{h,G_{L^{-1}}(\de\kappa)} \\
\leq& 2^{\abs{L^{-1}X_0}} \\
\leq & 2^{\abs{Z}}
\end{split} \ee 
Next, for the $K$ terms, they are also bounded by
\be\begin{split}
&\sup_{\phi,\zeta} \norm{K(X_i,\phi_L+\zeta_L)}_hG^{-1}(\kappa,X_i,\phi_L+\zeta_L)\\
\leq &\sup_{\phi} \norm{K(X_i,\phi_L)}_hG^{-1}(\kappa,X_i,\phi_L)\\
\leq &\norm{K(X_i)}_{h,G}\\
\end{split}\ee

Then the issue is to get a proper bound on $P(\De_j)$. First, the  formula \ref{reblockformula} converges only if $\norm{P(\De_j)}$ is sufficiently small. And under the stability condition, the good control on $P$ will give the desired bound on $K$. This will be the proof in \cite{BDH98} or \cite{BMS03}. But the proof is a bit lengthy, here we prove the bound with the specified choice of $V$  for simplification. Now suppose $V$ have the same form as in Lemma \ref{ve000}. Then for a single block $\De$,
\be
\begin{split}
P(\De,\phi_L,\zeta_L)&=e^{-V(\De,\phi_L+\zeta_L)}-e^{-V(\De,\phi_L)}\\
&=e^{-\sigma \int (\partial\phi_L+\partial\zeta_L)^2}-e^{-\sigma \int (\partial\phi_L)^2}\\
&=\int_0^1 \frac{d}{ds} e^{-\sigma \int (\partial\phi_L+s\partial\zeta_L)^2dx} ds
\end{split} 
\ee
Here $\sigma$ is a small quantity close to $0$, indeed we need $\sigma$ to be a $\delta\kappa$ dependent small term as shown in the following computations.
Then \be
\begin{split}
&\norm{P(\De,\phi_L,\zeta_L)}_h\\
\leq&  \norm{\int_0^1 \frac{d}{ds} e^{-\sigma \int (\partial\phi_L+s\partial\zeta_L)^2dx} ds}_h\\
\leq& \max_{0\leq s\leq 1}\norm{\frac{d}{ds}( -\sigma \int_{\De} (\partial\phi_L+s\partial\zeta_L)^2dx)}_h\cdot\norm{e^{-V(\De,\phi_L+s\zeta_L)}}_h\\
\leq& 
\norm{( -2\sigma \int_{\De} (\partial\phi_L\partial\zeta_L +(\partial\zeta_L )^2)dx)}_h
 \cdot\max_{0\leq s\leq 1}\norm{e^{-V(\De,\phi_L+s\zeta_L)}}_h
\end{split} \ee
and by the computations in Lemma \ref{ve000} and Cauchy-Schwarz inequality, the first term is bounded by
\[ \begin{split}
&\norm{( -2\sigma \int_{\De} (\partial\phi_L\partial\zeta_L +(\partial\zeta_L )^2)dx)}_h\\
\leq & \cO(1) \Big(\abs{\sigma}\norm{\zeta_L}_{H(\De)}^2 +\abs{\sigma}\norm{\phi_L}_{H(\De)}^2 +h  \abs{\sigma}\norm{\zeta_L}_{H(\De)} \\
 &\qquad+ h\abs{\sigma}\norm{\phi_L}_{H(\De)} +h^2\sigma^2  \Big)
\end{split}
\]
Now take the term $\norm{\zeta_L}^2_{H(\De)}$ for example. If we multiply by the corresponding regulator 
\be
\begin{split}
&G^{-1}(\delta\kappa,\Delta,\zeta_L)\\
=&\exp\left( -\delta\kappa\sum_{1\leq\abs{\alpha}\leq s} \int_{\mathring{\Delta}} \abs{\partial\zeta_L}^2-\de\kappa c\sum_{\abs{\alpha}=1} \int_{\partial\Delta} \abs{\partial\zeta_L}^2   \right)
\end{split}
\ee
to control  the growth, then according to the condition $\abs{\sigma}\leq\delta\kappa$, we have
\be
\abs{\sigma}\norm{\zeta_L}^2_{H(\De)} G^{-1}(\delta\kappa,\Delta,\zeta_L)\leq \cO(1)\abs{\sigma/\delta\kappa}
\ee
In the later sections, $\delta\kappa$ will be a sufficiently small $j-$dependent term, so is $\sigma$. 
Similarly for the other terms and for the term $\norm{\phi_L}_{H(\De)}\norm{\zeta_L}_{H(\De)}$.

And the $V$ term is also $\cO(1)$ bounded after multiplying the corresponding large field regulator $G(\De,\phi_L,\zeta_L)$ by similar argument as Lemma \ref{k0estimate}.
Therefore by multiplying the large field regulator of $\phi$ and $\zeta$,  we have:
\begin{align}
\norm{P(\De)}_{h,G}
\leq \cO(1)\abs{\sigma/\delta\kappa}
\end{align}
Then
\be \label{05312011}
\begin{split}
&\norm{\cS\cB K(Z)}_{h,\hat{G}_{L^{-1}}}\\
\leq& \cO(1)2^{\abs{Z}}\sum_{M,N}\frac{1}{N!M!}\sum_{(X_i),(\De_j)\rightarrow LZ } \prod_j \norm{K( X_i)}_{h,G}\prod_j\abs{ \sigma/\delta\kappa}
\end{split} \ee
Now note that by \cite{BDH98}, there is some constant $\cO(1)$ such that
\be
\Gamma_p(L^{-1}\bar{X}^L)\leq \cO(1) \Gamma_{p-q}(X)   \label{07132011}
\ee  
Therefore 
\be
\begin{split}
\Gamma_p(Z)\leq\Gamma_p(L^{-1}\overline{(\cup X_i)\cup(\cup \Delta_j)}^L) & \leq  \cO(1) \Gamma_{p-q}((\cup X_i)\cup(\cup \Delta_j))\\
&= \cO(1) \Gamma_{p-q}(\cup X_i)\Gamma_{p-q}(\cup \Delta_j)
\end{split}\ee
Then by  multiplying \ref{05312011} by large set regulator $\Gamma_{p}$  and taking the summation  over $\De$, we get 
\be
\begin{split}
&\norm{\cS\cB K}_{h,\hat{G}_{L^{-1}},\Gamma_{p}}\\
=&\sum_{Z\supset  \De}\Gamma_p(Z)\norm{\cB K(Z)}_{h,\hat{G}_{L^{-1}}} \\
\leq &\cO(1)\sum_{Z\supset  \De}\Gamma_{p+1}(Z )\sum_{M,N \atop M+N\geq 1}\frac{1}{N!M!}\\
&\qquad\qquad\sum_{(X_i),(\De_j)\rightarrow LZ}  \prod_i\norm{K( X_i)}_{h,G}\prod_j\abs{ \sigma/\delta\kappa}\\
\leq & \cO(1) \sum_{Z\supset  \De} \sum_{M,N \atop M+N\geq 1}\frac{1}{N!M!}\sum_{(X_i),(\De_j)\rightarrow LZ}  \prod_i\norm{K( X_i)}_{h,G}\prod_j\abs{ \sigma/\delta\kappa}\\ &\qquad\qquad\qquad\Gamma_{p-q+1}(\cup_i X_i)\Gamma_{p-q+1}(\cup_j\De_j)\\
\end{split}
\ee
Then by a spanning tree argument in \cite{BY90}, we have
\be
\begin{split}
&\norm{\cS\cB K}_{h,\hat{G}_{L^{-1}},\Gamma_{p}}\\
\leq &\cO(1)\sum_{N\geq 1}\cO(1)^NL^{2N}\left( \norm{K}_{h,G,\Gamma_{p-q+3}}+\abs{\sigma/\delta\kappa}\right)^N\\
\leq&\cO(1)L^2(\norm{K}_{h,G,\Gamma_{p-q+3}}+\abs{\sigma/\delta\kappa})
\end{split}
\ee
if $\norm{K}$ is sufficiently small.

\item 
First note that the corresponding regulators have the property that if $\kappa c^{-1} L^{2s-2}$ is sufficiently small, then
\be
\int G(\kappa,X,\phi+\zeta)d\mu_{C}(\zeta)\leq 2^{\abs{X}}G_L(\kappa,X,\phi)   \label{regulatorestimate}
\ee
For the proof of the estimate \ref{regulatorestimate}, see Appendix \ref{regulatorestimatesection}. Also See \cite{F10} Lemma 5.2 with a slightly different regulator, and some very similar results are given in Lemma 5 in \cite{DH00} and Lemma 6.10 in  \cite{B09}.
Now \ref{regulatorestimate} implies 
\be
\begin{split}
&\int \hat{G}(\kappa,X,\phi,\zeta)d\mu_C(\zeta)\\
=&\int G(\kappa,X,\phi+\zeta) G(\delta\kappa,X,\phi)G(\delta\kappa,X,\zeta)d\mu_C(\zeta)\\
\leq & G(\delta\kappa,X,\phi)\left( \int G^2(\kappa,X,\phi+\zeta)d\mu_C(\zeta)\right)^{\frac{1}{2}}\left( \int  G^2(\delta\kappa,X,\zeta)d\mu_C(\zeta)\right)^{\frac{1}{2}}\\
= & G(\delta\kappa,X,\phi) \left( \int G(2\kappa,X,\phi+\zeta)d\mu_C(\zeta)\right)^{\frac{1}{2}}\left( \int  G(2\delta\kappa,X,\zeta)d\mu_C(\zeta)\right)^{\frac{1}{2}}\\
\leq & G(\delta\kappa,X,\phi) \left( 2^{\abs{X}}G_L(2\kappa,X,\phi+\zeta)\right)^{\frac{1}{2}}\left(  2^{\abs{X}}\right)^{\frac{1}{2}}\\
\leq& G_L(\kappa+\delta\kappa,X,\phi)2^{\abs{X}}
\end{split} 
\ee
For the last two steps, we have used the fact that since $2\de\kappa \leq\kappa$,
\be\int G(2\de\kappa,X,\zeta)d\mu_C(\zeta)\leq\int G(\kappa,X,\zeta)d\mu_C(\zeta)\leq 2^{\abs{X}}G_L(\kappa,X,0)=2^{\abs{X}}
\ee
%
Then we consider the $\natural$ integration 
\be
\begin{split}
K^{\natural}(X,\phi)&=\int K(X,\phi,\zeta)d\mu_{C_{L^{-1}}}(\zeta)\\
&=\int K(X,\phi,\zeta_{L^{-1}})d\mu_C(\zeta)\\
&=\int K_L(LX,\phi_L,\zeta)d\mu_C(\zeta)\\
\end{split}\ee
where $K_L(LX,\phi_L)=K(X,\phi)$. After taking the functional derivatives, 
\be
\begin{split}
\norm{K^{\natural}(X,\phi)}_h
&\leq \int \norm{K_L(LX,\phi_L,\zeta)}_h d\mu_C(\zeta)\\
&\leq \norm{K_L(LX)}_{h,\hat{G}(\kappa)}\int \hat{G}(\kappa,LX,\phi_L,\zeta) d\mu_C(\zeta)\\
&\leq \norm{K(X)}_{h,\hat{G}_{L^{-1}}(\kappa)}2^{\abs{X}}G_L(\kappa+\delta \kappa,LX,\phi_L)\\
&\leq \norm{K(X)}_{h,\hat{G}_{L^{-1}}(\kappa)}2^{\abs{X}}G(\kappa+\delta \kappa,X,\phi)\\
\end{split}\ee
Here we have used the fact 
\[\begin{split}&\norm{K_L(LX)}_{h,\hat{G}(\kappa)}\\
=&\sup_{\phi,\zeta} \norm{K_L(LX,\phi_L,\zeta)}_h \hat{G}^{-1}(\kappa,LX,\phi_L,\zeta)\\
=&\sup_{\phi,\zeta'} \norm{K_L(LX,\phi_L,\zeta'_L)}_h \hat{G}^{-1}(\kappa,LX,\phi_L,\zeta'_L)\\
=&\sup_{\phi,\zeta'} \norm{K(X,\phi,\zeta')}_h \hat{G}^{-1}_{L^{-1}}(\kappa,X,\phi,\zeta')\\
=&\norm{K(X)}_{h,\hat{G}_{L^{-1}}(\kappa)}\end{split}\]
Then 
\be
\norm{K^{\natural}(X)}_{h,G(\kappa+\delta \kappa)}\leq 2^{\abs{X}}\norm{K(X)}_{h,\hat{G}_{L^{-1}}(\kappa)}
\ee
Therefore,
\be
\norm{K^{\natural}}_{h,G(\kappa+\delta\kappa),\Gamma_{p}} \leq \norm{K}_{h,\hat{G}_{L^{-1}}(\kappa),\Gamma_{p+1}}
\ee
The change of $\Gamma_p\rightarrow \Gamma_{p+1}$ is the effect of the $2^{\abs{X}}$ factor. 

Similarly for $K^{\#}(X,\phi)=\int K(X,\phi+\zeta)d\mu_C(\zeta)$, we have 
\be \begin{split}
&\norm{K^{\#}(X)}_{h,G_L(\kappa+\de\kappa)} \\
\leq&\sup_{\phi}  \int\norm{K(X,\phi+\zeta)}_h  G_L^{-1}(\kappa+\de\kappa,X,\phi)  d\mu_C (\zeta)      \\
\leq&\sup_{\phi}  \int\norm{K(X,\phi+\zeta)}_h   \hat{G}^{-1}(\kappa,\de\kappa,X,\phi,\zeta)G_L^{-1}(\kappa+\de\kappa,X,\phi)\\
 &\qquad\qquad \hat{G}(\kappa,\de\kappa,X,\phi,\zeta)d\mu_C (\zeta)      \\
\leq&\sup_{\phi} \norm{K(X)}_{h,\hat{G}} \int\hat{G}(\kappa,\de\kappa,X,\phi,\zeta) G_L^{-1}(\kappa+\de\kappa,X,\phi)d\mu_C (\zeta) \\
\leq & 2^{\abs{X}}\norm{K(X)}_{h,\hat{G}}
\end{split}\ee
In the last step follows from a consequence of \ref{regulatorestimate}: 
\be
\begin{split}
&\int \hat{G}(\kappa,X,\phi,\zeta)d\mu_C(\zeta)\\
=&\int G(\kappa,X,\phi+\zeta) G(\delta\kappa,X,\phi)G(\delta\kappa,X,\zeta)d\mu_C(\zeta)\\
\leq & G(\delta\kappa,X,\phi)\left( \int G^2(\kappa,X,\phi+\zeta)d\mu_C(\zeta)\right)^{\frac{1}{2}}\left( \int  G^2(\delta\kappa,X,\zeta)d\mu_C(\zeta)\right)^{\frac{1}{2}}\\
= & G(\delta\kappa,X,\phi) \left( \int G(2\kappa,X,\phi+\zeta)d\mu_C(\zeta)\right)^{\frac{1}{2}}\left( \int  G(2\delta\kappa,X,\zeta)d\mu_C(\zeta)\right)^{\frac{1}{2}}\\
\leq & G(\delta\kappa,X,\phi) \left( 2^{\abs{X}}G_L(2\kappa,X,\phi+\zeta)\right)^{\frac{1}{2}}\left(  2^{\abs{X}}\right)^{\frac{1}{2}}\\
\leq& G_L(\kappa+\delta\kappa,X,\phi)2^{\abs{X}}
\end{split} 
\ee
\end{enumerate}
\end{proof}
The above lemma is about the changes of polymer activities under a single renormalization group step.
Since the R.G. Transformation $K_j\rightarrow K_{j+1}$ is composed by all of the above transformations formally, we obtain the following result.
\begin{lem} \label{07182011}
 Under the  conditions as Lemma \ref{rges},
 \be \label{nrges}
 \begin{split}
 &\norm{\cE((\cS\cB K)^{\natural},F_{0},F_{1})}_{h,G(\kappa+\delta\kappa),\Gamma}\\
 \leq& \cO(1)\left(L^2\norm{K}_{h,G(\kappa),\Gamma}+\norm{\de f}_{\Gamma_5}+L^2\abs{\sigma/\de\kappa}\right)
 \end{split}
 \ee
\end{lem}
\begin{proof}
The bound follows by combining the bounds in Lemma \ref{rges} together:
\begin{align}
&  \norm{\cE((\cS\cB K)^{\natural},F_{0},F_{1})}_{h,G(\kappa+\delta\kappa),\Gamma_p}\\
\leq& \cO(1)\Big(  \norm{(\cS\cB K)^{\natural}}_{h,G(\kappa+\delta\kappa),\Gamma_{p+3}}+\norm{\de f}_{\Gamma_{p+5}} \Big)\\
\leq& \cO(1)\Big(  \norm{\cS\cB K}_{h,\hat{G}_{L^{-1}}(\kappa),\Gamma_{p+4}}+\norm{\de f}_{\Gamma_{p+5}} \Big)\\
\leq& \cO(1)\Big( L^2 \norm{K}_{h,G(\kappa),\Gamma_{p+7-q}}+L^2\abs{\sigma/\de\kappa}+\norm{\de f}_{\Gamma_{p+5}} \Big)\\
\leq&\cO(1)\left(L^2\norm{K}_{h,G(\kappa),\Gamma_p}+L^2\abs{\sigma/\de\kappa}+\norm{\de f}_{\Gamma_{p+5}}\right)
\end{align}
by choosing $q=7$.
\end{proof}

\section{R.G. Analysis on Partition Function}
During each renormalization transformation, the extraction is made from the relevant and marginal terms of polymer activities. In \cite{DH00}, those extraction terms are absorbed into the Gaussian measure in each step. In the present paper, we are using finite range decomposition covariance, where the absorption technique is no longer available.

\subsection{Initial Polymer Activities}
Start with the definition (\ref{pdef})
\be 
Z(\La_M,z,\sigma)=\int \exp \left( zW(\La_M,\phi)-\sigma V(\La_M,\phi) \right) d\mu_{\beta v_M}(\phi)
\ee
and let $\Delta$ denote unit blocks in $\La_M$,  the the idea is to break the polymer activity into pieces
\begin{align}
&\exp \left(zW(\La_M)-\sigma V(\La_M)\right) \\
=&\prod_{\Delta\subset\La_M} \exp \left( zW(\Delta)-\sigma V(\Delta)\right)\\
=&\prod_{\Delta\subset\La_M} \left[  \left(e^{zW(\Delta)}-1 \right)e^{-\sigma V(\Delta)}   +e^{-\sigma V(\Delta)}  \right] \\
=&\sum_{X\subset \La_M} \left(  \prod_{\Delta\in X}  \left(e^{zW(\Delta)}-1 \right)e^{-\sigma V(\Delta)} \right)   \cdot\left( \prod_{\Delta\notin X} e^{-\sigma V(\Delta)}\right)\\
=&\sum_{X\subset \La_M} K_0(X) e^{-\sigma V(\La_M \setminus X)}\\
=&\cE xp\left(\Box e^{-V_0}+K_0\right)(\La_M,\phi)
\end{align}
where $V_0(X)=\prod_{\De\subset X}V_0(\De)$ with
\be
V_0(\Delta)=\sigma V(\Delta)=\sigma \int_{x\in\Delta} (\partial\phi)^2dx
\ee
and for connected polymer $X\subset\La_M$,
\be
K_0(X)=\prod_{\Delta\in X}  \left(e^{zW(\Delta)}-1 \right)e^{-V_0(\Delta)}
\ee
Note that both $V_0$ and $K_0$ have factorization property over unit blocks and $\phi-$localization property on unit blocks. In the dipole gas model \cite{D09}, $K_0$ even have the translation invariant due to the original $\partial\phi$ potential $V$  which advantage we don't have here. And this is also the essential difference between them. 

\begin{lem} \label{k00estimate}
For sufficiently small $z$ and $\sigma$, and $0<\epsilon<1$
\be \label{k0estimate}
\norm{K_0}_{h,G,\Gamma} \leq \cO(1)\abs{z}^{1-\epsilon}
\ee
\end{lem}
\begin{proof}
\mbox{}

First note that the $n-$th derivatives of $W$ satisfy $\|W_n(\De,\phi) \| \leq 1 $. Then we
take  the $n^{th}$ functional derivatives and re-sum  to  find   
\be \frac {(2h)^n}{n!}
\|(e^{z W(\De)  })_n(\phi)\|
\leq   \exp  \left( \sum_{n=0}^{\infty} \frac {(2h)^n}{n!}\abs{z}
\norm{W_n(\De,\phi)}  \right)  \ee
Multiply  both sides by $2^{-n}$ and sum over $n$, we get
\be
\begin{split}
\norm{e^{zW(\De)}}_{h}\leq&\sum_{n=0}^{\infty}\frac{1}{2^n}\exp  \left( \sum_{n=0}^{\infty} \frac {(2h)^n}{n!}\abs{z}
\norm{W_n(\De,\phi)}  \right)\\
\leq & 2\exp\Big(\abs{z}e^{2h}\Big)
\end{split}\ee
Then by taking the supremum over $\phi$,
\be
\norm{e^{z W(\De)}}_{h,G=1}\leq 2\exp\Big(\abs{z}e^{2h}\Big)
\ee
Now we write
\be   e^{ z W(\De)}  -1  = \frac{1}{2\pi i}  \int \frac {e^{z \zeta W(\De) }}{\zeta(\zeta-1)}   d\zeta  \ee
where the contour is the circle  $ |\zeta| = |z|^{-1+\epsilon/2} \geq 2$.
Since   $\| e^{z\zeta  W(\De) }\|_{h,1}  \leq \cO(1)$
for $\abs{\zeta}$ small
and 
 $\cO(1) |z|^{1-\epsilon/2} \leq |z|^{1-\epsilon}$, so we have that for $z$ sufficiently small:
\be
\norm{e^{zW(\Delta)}-1}_{h,G=1}\leq \abs{z}^{1-\epsilon}
\ee
And for the other term, we have for $\sigma$ sufficiently small:
\be
\norm{e^{-V_0(\Delta)}}_{h,G}\leq 2
\ee
The proof of this bound is included in Lemma \ref{06152011} as a special case.
Combine them by the norm property (\ref{normproperty}), we will have
\be
\norm{K_0(\De)}_{h,G}\leq \norm{e^{zW(\Delta)}-1}_{h,G=1}\norm{e^{-V_0(\Delta)}}_{h,G} \leq \cO(1)\abs{z}^{1-\epsilon}
\ee
Finally, by the factorization property of $K(X)=\prod_{\De\subset X}K(\De)$ and the standard regulator estimate, (\ref{k0estimate}) holds. 
\end{proof}
The lemma shows the initial polymer activity could be sufficiently small, 
which gives us the necessary condition to continue the R.G. analysis to the next scales.

\subsection{R.G. Analysis} 
Now we expect to set up the renormalization group transformation such that
\be
\begin{split}
&Z(\La_M,z,\sigma)\\
=&e^{\sum_j\sum_{X\subset \La_{M-j-1}} F_j(X)}\int \cE xp(\Box e^{-V_j} +K_j)(\La_{M-j},\phi)  d\mu_{\beta v_{M-j}}(\phi)
\end{split}
\ee
holds for any $0\leq j\leq M-1$.
Then we need a good control for the growth  of the polymer from activities $(V_j,K_j)$ to  $(V_{j+1},K_{j+1})$. The order is: given $(V_j,K_j)$, we want to make a specified choice of $V_{j+1}$ via the extraction, then $K_{j+1}$ is determined by $K_j$ and $V_{j+1}$.

Although our RG construction is for $(V_j,K_j)$, it is convenient to assume $V_j$ is a function of a coupling constant $\sigma_j$ to fit in the tuning program of the coupling constants. More precisely, we  pre-assume that $V_j(\De,\phi)=\sigma_j \int_{x\in\De} (\partial \phi)^2 dx$ holds for any $j$.  Therefore for each single RG transform, we regard the map $(V_j,K_j)\rightarrow (V_{j+1},K_{j+1})$ is the same as the map  $(\sigma_j,K_j)\rightarrow (\sigma_{j+1},K_{j+1})$.

\subsubsection{Estimates on $V$} \label{VEstimates}
Roughly speaking, the extraction in each step is made from $(\cS\cB K)^{\natural}$. And $K_{j+1}$ is determined by both $V_{j+1}$(which indeed is also specified by the extraction) and the extraction. 

More precisely,
for $\cE xp(\Box e^{-V}+K)$,
the extracted part $F=F_0+F_1$ is from neutral sector 
\be\overline{K}(X,\phi)=\frac{1}{2\pi}\int_{0}^{2\pi}K(X,\Phi+\phi)d\Phi\ee
for small sets $X$ i.e. the size of $X$ is no larger than $2^2$, which is defined in Section \ref{fundmentaltransform}.
The extraction satisfies:
\be
\begin{cases}
F(X,\phi+c)=F(X,\phi)\\
\dim\left(\overline{K}-F\right)\geq 4 \\
\end{cases}
\ee
where the $\dim$ is defined in Section \ref{KEstimates} which is originally from the neutral argument of \cite{DH00}. Those conditions can be summarized to small sets $X$:
\be \label{06302011b}
\begin{cases}
\left(  \overline{K}-F \right)_0(X,0)=0\\
\left(  \overline{K}-F \right)_2(X,0;x_{\mu},x_{\nu})=0\\
\left(  \overline{K}-F \right)_2(X,0;x_{\mu},x_{\nu}x_{\rho})=0\\
\end{cases}
\ee

Let 
\be \label{06272011}
F(X,\De,\phi)=F_0(X,\De,\phi)+F_1(X,\De,\phi)\ee
with
\begin{align}
F_0(X,\De,\phi)&= \alpha^{(0)}(X)\\
F_1(X,\De,\phi)&= \sum_{\mu,\nu}\alpha_{\mu,\nu}^{(2)}(X)\int_{\De}(\partial_{\mu}\phi)(\partial_{\nu}\phi)
+\sum_{\mu,\nu\rho}\alpha_{\mu,\nu\rho}^{(2)}(X)\int_{\De}(\partial_{\mu}\phi)(\partial_{\nu\rho}^2\phi)
\end{align}

By the above conditions, we determine those coefficients to be:
\be   
\begin{cases}
\alpha^{(0)}(X)=\frac{1}{|X|}\overline{ K}_0(X,\phi=0) \cdot\chi_{\cS}(X) \\  
\alpha^{(2)}_{\mu,\nu}(X)=\frac{1}{2|X|}\overline{ K}_2(X,\phi=0;x_{\mu},x_{\nu}) \cdot\chi_{\cS}(X) \\
\alpha^{(2)}_{\mu,\nu\rho}(X)=\frac{1}{2|X|}\overline{ K}_2(X,\phi=0;x_{\mu},x_{\nu}x_{\rho}) \cdot\chi_{\cS}(X) \\
\end{cases}
\ee
Here $\chi_{\cS}$ is the characteristic function of small sets.

And by the translation invariant property of the polymers $X$, 
we can show that if the summation  is over all polymers 
then those coefficients satisfies:
\be   \label{06152011b}
\begin{cases}
\sum_{X\supset\De}\alpha^{(0)}(X)=\de E \\  
\sum_{X\supset\De}\alpha^{(2)}_{\mu,\nu}(X)=-\frac{1}{2\beta}\de_{\mu\nu}\de\sigma \\
\sum_{X\supset\De}\alpha^{(2)}_{\mu,\nu\rho}(X)=0 \\
\end{cases}
\ee
where $E$ and $\sigma$ are constants. 
Here the point is both $\de E$ and $\de_{\mu\nu}\de\sigma$ should be $\De-$independent.  

Therefore,
\begin{align}
V_{F_0}(\De)&=\sum_{X\supset\De} F_0(X,\De) = \de E\\
V_{F_1}(\De)&=\sum_{X\supset\De} F_1(X,\De)=-\frac{\de\sigma}{2\beta}\int_{x\in\De} (\partial\phi)^2    \label{06272011b}
\end{align}

\begin{lem} \label{06152011}
Let \[V(\De,\phi)=\sigma\int_{\Delta}(\partial\phi(x))^2dx\]
Then for $\sigma/\kappa$ sufficiently small, we have
\be
\norm{e^{-V}}_{h,G(\kappa),\Gamma} \leq 2
\ee
which also verifies the hypothesis \ref{06062011}.
\end{lem}
\begin{proof}
The bound for $V$ follows from direct calculation. 
\begin{align}
&V_n(\De,\phi;f_1,\cdots,f_n)\\
=&\frac{\de^n}{\de t_1\cdots \de t_n}V(\De,\phi+t_1f_1+\cdots t_nf_n)\Big|_{t=0}\\
=&\sigma\frac{\de^n}{\de t_1\cdots \de t_n}\int_{x\in\De}\partial(\phi+t_1f_1+\cdots t_nf_n)(x)^2dx\Big|_{t=0}\\
=& 
\begin{cases} 
2\sigma\int_{\De} \partial\phi(x)\partial f_1(x)dx     &\text{if $n=1$}    \\
2\sigma\int_{\De} \partial f_1(x)\partial f_2(x)dx     &\text{if $n=2$}    \\
0&\text{if $n\geq 3$}
\end{cases}
\end{align}
which implies
\be
\norm{V_n(\De,\phi)}\leq
\begin{cases} 
2\abs{\sigma}\norm{\phi}_{H(\De)}    &\text{if $n=1$}    \\
2\abs{\sigma}   &\text{if $n=2$}    \\
0&\text{if $n\geq 3$}
\end{cases}
\ee
Then by summing over $n$,
\be \begin{split} \label{06232011}
\norm{V(\De,\phi)}_h &\leq\abs{\sigma}\norm{\phi}_{H(\De)}^2+2h\abs{\sigma}\norm{\phi}_{H(\De)}+h^2\abs{\sigma} \\
&\leq 2\abs{\sigma}(\norm{\phi}_{H(\De)}^2+h^2)
\end{split} \ee
Then the $n^{th}-$derivative of $e^{-V}$ satisfy the bound
\be
\frac{(4h)^n}{n!}\norm{(e^{-V(\De)})_n(\phi)}\leq \exp\left(\sum_{n=0}^{\infty}\norm{V_n(\De,\phi)} \right)
\ee
Then multiply both sides by $4^{-n}$ and sum over $n$
\be \begin{split}
\norm{e^{-V(\De,\phi)}}_h&=\sum_{n=0}^{\infty} \frac{h^n}{n!}\norm{(e^{-V(\De)})_n(\phi)}\\
&\leq\sum_{n=0}^{\infty} \frac{1}{4^n}\exp \left( 4^2\norm{V(\phi,\De)}_h   \right)
\end{split}\ee
By taking the supremum over $\phi$,
\be  \label{06132011b}\begin{split}
\norm{e^{-V(\De)}}_{h,G}&=\sup_{\phi}\norm{e^{-V(\De,\phi)}}_h G^{-1}(\kappa,\De,\phi)\\
&\leq\sup_{\phi}\sum_{n=0}^{\infty} \frac{1}{4^n}\exp \left( 4^2\norm{V(\phi,\De)}_h   \right)G^{-1}(\kappa,\De,\phi)
\end{split}\ee
Now by \ref{06232011}, we have
\be \begin{split}
 &4^2\norm{V(\phi,\De)}_h -\log G(\kappa,\De,\phi)\\
 \leq &32 \abs{\sigma}h^2+(32\abs{\sigma}-\kappa) \norm{\phi}^2_{H(\De)}\\
 \leq & 32 \abs{\sigma}h^2
\end{split}\ee
if $32\abs{\sigma}-\kappa<0$. Also note that  $\cO(1)h^2\leq\frac{1}{\kappa}\leq \cO(1)h^2$. Therefore, if $\abs{\sigma} h^2\leq \frac{1}{32}\log \frac{3}{2}$ and $\frac{\abs{\sigma}}{\kappa}\leq 1/32$, i.e, if $\abs{\sigma}h^2$ or $\abs{\sigma}/\kappa$ is sufficiently small, then
\be
\norm{e^{-V(\De)}}_{h,G}\leq \exp(32 \abs{\sigma}h^2)\sum_{n=0}^{\infty} \frac{1}{4^n} \leq 2
\ee
which implies the claimed bound.
\end{proof}

\begin{lem} \label{ve000}
\mbox{}
Under  the extraction specified (\ref{06272011})-(\ref{06272011b}), RG transform will be
\be
\int\cE xp(\Box e^{-V}+K)(L\La,\phi_L+\zeta)d\mu_C(\zeta)=e^{\de\cE}\cE xp(\Box e^{-V''}+K')(\La,\phi)
\ee
Moreover,
\begin{enumerate}
\item Suppose that 
\be V(\Delta,\phi)=\sigma\cdot\int_{x\in\Delta} (\partial\phi)^2dx\ee
where $\sigma$ is a positive constant.
Then under such extraction, we can choose
\begin{align}
(\de\cE)(\Delta,\phi)&=\de E\\
V''(\Delta,\phi)&=\sigma'\cdot\int_{x\in\Delta} (\partial\phi)^2dx
\end{align}
 where $\de E$ and $\sigma'=\sigma-\delta\sigma/2\beta$ are  constants defined in \ref{06152011b}.
\item The extraction we made here is stable in the sense of \ref{stabilitycondition} if  $\sigma/\kappa$ and $\sigma'/\kappa$ are sufficiently small. 
\item 
For the quantities, we  have:
\begin{align}
\abs{\de E}&\leq \cO(1)\norm{K}_{h,G,\Gamma}\\
\abs{\de \sigma}&\leq \cO(1)\norm{K}_{h,G,\Gamma}
\end{align}
\end{enumerate}
\end{lem}

\begin{proof}
\mbox{}
\begin{enumerate}
\item

For $j=0$, the result is trivial with $E_0=0$.

Note that $V''$ depends on the choice of $V$, besides the extraction $V_{F}$. Now start with
\be
V(\Delta,\phi)=\sigma\cdot\int_{x\in\Delta} (\partial\phi)^2dx
\ee
First, after  fluctuation (\ref{flucformula}) and rescaling (\ref{scaleformula})steps, we will have
\be\begin{split}
&\int \cE xp\left(\Box e^{-V}+K\right)\left(L\La,\phi_L+\zeta\right)d\mu_C(\zeta) \\
=& \cE xp\left(\Box e^{-V_{L^{-1}}}+ (\cS \cB K)^{\#}  \right)\left(\La,\phi\right)
\end{split}
\ee
where
$V_{L^{-1}}(\De,\phi)=V(L\De,\phi_L)$
still keeps the standard form.

Combine with the extraction (\ref{extraformula0}) which is specified right above the Lemma \ref{ve000}, $V'$ will have the standard form:
\be \begin{split}
V'(\De,\phi)&=V_{L}(\De,\phi)-V_F(\De,\phi)\\
&=V(L\De,\phi_L)-(\de E-\frac{\de \sigma}{2\beta}\int (\partial \phi)^2)\\
&=E'(\Delta)-\sigma'\cdot\int_{x\in\Delta} (\partial\phi)^2dx
\end{split}
\ee
where $E'=E-\de E$ and $\sigma'=\sigma-\de \sigma/2\beta$  are constants. Then we extract the constant energy term $E$ to get the final the extraction formula
 (\ref{extraformula}) with:
 \begin{align}  
 V''(\Delta,\phi)&=(\sigma-\de\sigma)\cdot\int_{x\in\Delta} (\partial\phi)^2dx \\
 \de\cE&=\sum_{X\subset L\La} F_0(X)=\de E\abs{L\La} \label{07112011}
 \end{align}
where $\de E$ is a constant depending on the unit blocks only.
 
\item 
Now we  check the stability condition \ref{stabilitycondition} of such extraction. The verification is analogous to part of the proof of Theorem 18 in \cite{DH00}. Suggested by Lemma 21 of \cite{DH00}, extraction $F=F_0+F_1$ is stable for $(h,G,\de f(X))$ in  \ref{stabilitycondition} if we define 
\be\de f(X)=  80k(\abs{\alpha^{(0)}(X)}+ \delta\kappa^{-1} \sum_{\mu,\nu}\abs{\alpha^{(2)}_{\mu\nu}(X)}  +\delta\kappa^{-1} \sum_{\mu,\nu\rho}\abs{\alpha^{(2)}_{\mu\nu\rho}(X)}  )         \label{05292011}     \ee
where $k=\cO(1)$ is the number of small sets containing a unit block $\De$.

Then as the proof in Lemma \ref{k00estimate}, we have
\be
\begin{split}
&\frac{(3h)^n}{n!}\norm{\exp\left( - V(\De)-\sum_{X\supset\De} z(X)F(X,\De) \right)_n(\phi)}\\
\leq& \exp\left(\sum_{n=0}^{2}\frac{(3h)^n}{n!} \norm{V_n(\De,\phi)}  +\sum_{X\supset \De} \abs{z(X)}\sum_{n=0}^2 \frac{(3h)^n}{n!}\norm{F_n(X,\De,\phi)} \right)\\
\leq& \exp\left( \norm{V_n(\De,\phi)}_h  +\sum_{X\supset \De} \abs{z(X)}\sum_{n=0}^2 \frac{(3h)^n}{n!}\norm{F_n(X,\De,\phi)} \right)\\
\leq& \exp\Big(\cO(1)\abs{\sigma}( h^2 + \norm{\phi}^2_{H(\De)} )+\sum_{X\supset \De} \abs{z(X)}\sum_{n=0}^2 \frac{(3h)^n}{n!}\norm{F_n(X,\De,\phi)} \Big)
\end{split}
\ee
And those derivatives in $F$ are bounded by
\begin{align}
\norm{F_0(X,\De,\phi)}&\leq \abs{\alpha^{(0)}(X)}+ \abs{\alpha^{(2)}(X)} \norm{\phi}^2_{H(\De)}\\
\norm{F_1(X,\De,\phi)}&\leq 2\abs{\alpha^{(2)}(X)} \norm{\phi}_{H(\De)}\\
\norm{F_2(X,\De,\phi)}&\leq 2\abs{\alpha^{(2)}(X)} 
\end{align}
which implies the bound
\be \begin{split}
&\sum_{X\supset \De} \abs{z(X)}\sum_{n=0}^2 \frac{(3h)^n}{n!}\norm{F_n(X,\De,\phi)}\\ 
\leq &\sum_{X\supset \De} \abs{z(X)}(\abs{\alpha^{(0)}(X)}+(9h^2+6h \norm{\phi}_{H(\De)}+\norm{\phi}^2_{H(\De)})\abs{\alpha^{(2)}(X)})\\
\leq &\sum_{X\supset \De} \abs{z(X)}(\abs{\alpha^{(0)}(X)}+(10h^2+10 \norm{\phi}^2_{H(\De)})\abs{\alpha^{(2)}(X)})\\
\leq &\sum_{X\supset \De} \abs{z(X)}(\abs{\alpha^{(0)}(X)}+10\de\kappa^{-1}(1+ \kappa\norm{\phi}^2_{H(\De)})\abs{\alpha^{(2)}(X)})\\
\leq &\sum_{X\supset \De} k^{-1}(\frac{1}{8}+\frac{\kappa}{2}\norm{\phi}^2_{H(\De)})\de f(X)\abs{z(X)}\\
\leq &\frac{1}{8}+\frac{\kappa}{2}\norm{\phi}^2_{H(\De)}
\end{split}\ee
Here we used  $\de\kappa<\kappa=\cO(h^{-2})$ and $\de f(X)\abs{z(X)}\leq 1 $. 

Then again by a similar argument as Lemma  \ref{k00estimate}, 
\be \begin{split}
&\norm{\exp\left( -V(\De)-\sum_{X\supset\De} z(X)F(X,\De)\right)}_{h,G}\\
\leq & \sup_{\phi} \sum_{n=0}^{\infty}\frac{1}{3^n} \exp \Big(  3^2\norm{ -V(\De,\phi)-\sum_{X\supset\De} z(X)F(X,\De,\phi)  }_h \Big) G^{-1}(\kappa, X,\phi)\\
\leq &\sum_{n=0}^{\infty}\frac{1}{3^n} \sup_{\phi} \exp\Big(\frac{1}{8}+\cO(1)\abs{\sigma} h^2+ (\cO(1)\abs{\sigma}+\frac{\kappa}{2}-\kappa)\norm{\phi}^2_{H(\De)}   \Big) \\
\leq & \exp\Big(\frac{1}{8}+\cO(1)\abs{\sigma} h^2  \Big) \sum_{n=0}^{\infty}\frac{1}{3^n}\\
\leq &4
\end{split}
\ee
if $\cO(1)\abs{\sigma} h^2$ is sufficiently small.
Now we need still to check $\norm{\de f}$ . By the construction of the extraction, $\de f$ depends on the corresponding $K$. More precisely,  if $1\leq \kappa^{-1} h^{-2}\leq 4$ 
, then each term in \ref{05292011} is bounded by $\norm{K}$ which implies the estimate
\be
\norm{\de f}_{\Gamma_p}\leq\cO(1)\norm{K}_{h,G,\Gamma_{p}} \label{07122011}
\ee
Here the parameter $p$ is not important  because of the arbitrary choice of the parameter  $q$ on small sets.

If $\norm{K}_j$ is sufficient small(which is  true by the following sections), so is $\norm{\de f}$. 

Therefore the extraction we made here is stable in the sense of \ref{stabilitycondition}.
 
\item 
By \ref{06152011b}, the bound of $\de E$ and $\de \sigma$ actually depend on the bound of $K$, then the bounds follow from
\begin{align}
\abs{\de E}&\leq \sum_{X\supset \De} \norm{\alpha^{(0)}(X)}_{h,{G(\kappa)}}\Gamma(X)\leq \cO(1)\norm{K}_{h,G,\Gamma}\\
\abs{\de \sigma}&\leq \sum_{X\supset \De} \norm{\alpha^{(2)}(X)}_{h,G(\kappa)}\Gamma(X)\leq \cO(1)\norm{K}_{h,G,\Gamma}
\end{align}
Later we will also see that $K$ will also be sufficiently small quantities depending on activity $\zeta$.

\end{enumerate}

\end{proof}
Note that when the polymer activity $K$ is replaced by $\cS\cB K$ which is defined on $L^{-1}\La$, then the above summations will be over the next scale $L^{-1}\La$, especially for the term $\de\cE$ in \ref{07112011}.

\subsubsection{Estimates on $K$} \label{KEstimates}

Let  $K$ be a polymer activity which satisfies
$K(X, \phi + 2 \pi ) =  K(X, \phi  ) $.  Expand
    $  K(X,  \Phi + \phi ) $  in a Fourier series in the
real variable $\Phi$:
\be  
K(X, \Phi + \phi)= k_0(X,\f)+ \sum_{q\ne 0}  e^{iq\Phi} k_q(X,\phi)
\ee 
where for each $q$,
\be 
k_q(X,\phi)  = \frac{1}{2\pi} \int_{-\pi} ^{\pi} e^{-iq\Phi} K(X, \Phi + \phi) d \Phi  \label{06282011}
\ee
Then
\be  K(X, \phi)= k_0(X,\f)+ \sum_{q\neq 0}  k_q(X,
\phi)
\label{series}
\ee
The terms with $q\neq 0$ are called the  charged terms and the $q=0$
term is called the  neutral term. 

 Note that for any constant  $c$, we have
\be
k_q (X, \phi + c ) = e^{iqc} k_q(X, \phi) \ee
in particular, for the neutral sector,
\be
k_0(X,\phi+c)=k_0(X,\phi)
\ee

We also define the scaling dimension $\dim K$  of any polymer activity $K$
by
\begin{align}
        \dim ( K_{n} ) &= r_n + n \dim \phi\\
        \dim  (K)&= \inf_{n} \dim ( K_{n})
\label{dimK}
\end{align}
where the infimum
is taken over $n$ such that $K_n(X,0)  \neq 0$. Here $r_n$ is defined to
be  the
largest integer satisfying $r_n \leq r$ and
$K_{n}(X,\phi=0;p^{\times n})=0$ whenever
$p^{\times n}=(p_1,\dots,p_n)$ is an $n$--tuple of polynomials of total
degree less than $r_n$.  One can interpret $r_n$ as the
number of derivatives present in the  $\phi^n$ part of $K$  (up to a
maximum  $r$). Here the neutrality $K(X,\phi+c)=K(X,\phi)$ condition implies $K_n(X, \phi;f_1,...,f_n)$
vanishes if
any $f_i$ is a constant.  Hence for neutral $K$, 
$ \dim K_n =r_n \geq n $ for  $n <r$ and  $ \dim K_n =r_n=r $ for  $n 
\geq r$.

Then we have the following estimates:
\begin{lem} \label{e0051}
Let   $K(X, \phi)$ be supported on small sets,
and be  periodic in $\phi$.
\begin{enumerate}
\item For $q \neq 0$
\be  
\norm{\mu_C * k_q }_{h,G_{L}(\kappa),  \Gamma_{-1}} \leq m_q  \norm{ k_q }_{h + N_C,G(\kappa),  \Gamma}  
\ee
where
\begin{align}   
 N_C &=  \sup_{X\ \mbox{\rm small}}\  \inf_{x\in X}\norm{C(\cdot -
x)-C(0)}_{C^r(X)}=\cO(1)\\ m_q  &=   \exp [ -( |q| - 1/2) C(0) ]. 
\end{align}

\item If $k_q$ is supported on $L-$polymers $Y$, then for $0\leq \eta \leq 1$ ,
\be  
\norm{\cS k_q(L^{-1}Y)}_{h,G_{L^{-1}}(\kappa)}\leq \cO(1) e^{ \eta h|q|}  \norm{k_q(Y)}_{h(1-\eta/2),G(\kappa)}
\ee
\item If $k_0$ is supported on $L-$polymers $Y$, then
\be  \norm{ \cS k_0(L^{-1}Y)}_{h,G_{L^{-1}}(\kappa)}
\leq  \cO(1) L^{-  \dim (k_0)}   \norm{ k_0(Y)}_{h,G(\kappa)}   \ee
\end{enumerate}
\end{lem}
\begin{proof}
\mbox{}
\begin{enumerate}
\item The proof is essentially the same as Lemma 13 in \cite{DH00} except the covariance. The point here is for $X$ small, the quantity $N_C$ is always nonzero for finite range $C$, which implies the loss of the analyticity still exists.

By shifting
the integral by
$\zeta \to  \zeta  + i \si_q C(\cdot-x)$,  we will have
\begin{align}
	&\mu_C * k_q (X,\phi)\\
	=&\int k_q(X,\phi+\zeta)d\mu_C(\zeta)\\
	=&e^{C(0)/2}\int e^{-i\sigma_q \zeta(x)} k_q(X,  \phi + \zeta + i\sigma_q  C( \cdot - x) )d  \mu_{ C}(\zeta)	\\
	=&m_q \int e^{-i\sigma_q \zeta(x)}\underbrace{ k_q(X,  \phi + \zeta + i\sigma_q  (C( \cdot - x)-C(0) ))}_{\mbox{say } k_{q,x}(X,\phi+\zeta)}d  \mu_{ C}(\zeta)
\end{align}
Then by taking the functional derivatives and norms,
\be
\norm{(\mu_C*k_q)_n(X,\phi)}\leq m_q\int \norm{(k_{q,x})_n(X,\phi+\zeta)} d\mu_C(\zeta)
\ee
Since for finite range covariance regulators, the estimate
\be
\mu_C*G(\kappa,X)\leq G_L(\kappa,X)2^{\abs{X}}
\ee
is still true.
So the result follows by combining this with the estimate of the norm property:
\be\norm{K(\cdot+if)}_{h,G}\leq \norm{K(\cdot)}_{h+\norm{f}_{\cC^r(X)},G}\ee

To check $N_C=\cO(1)$, it suffices to check $\inf_x \sup_y \abs{C(x-y)-C(0))}$ by Lemma \ref{dimock0311}. then it follows from
\be \begin{split}
\abs{C(x-y)-C(0)}\leq & \int_1^L\frac{1}{l}\abs{u(\frac{x-y}{l})-u(0)} dl\\
\leq & \int_1^L\frac{1}{l^2} dl\cdot\max_{\abs{r}\leq 1} \abs{u'(r)}\leq const\end{split}
\ee
Here we have used the fact that $\abs{x-y}\leq \cO(1)$ since the set $X$ is small.

The key factor in the Lemma is $C(0)$ which is $1/2\pi\cdot\log L$. The summation over this factor will beat the $L^2$ factor. See the proof of Theorem \ref{RGcontrol} for details.

\item This is similar to Lemma 14 in \cite{DH00}.
Recall that
$\cS k_q(L^{-1}Y,\phi)=k_q(Y,\phi_L)$
Now we shift $\phi_L$  by a constant $\eta\phi_L(y_*)$ where $y_*$ is an arbitrary point of $Y$, then we have
\be    \cS k_q(L^{-1}Y,\phi)=  e^{iq \eta \phi_L(y_*) } k_{q}(Y, (1-\eta) \phi_L + \eta \tilde\phi_L ) \ee
 Here we have defined   $\tilde f(x) = f(x) - f(y^*/L)$ so that
$\tilde f_L(y)  = f_L(y) - f_L(y^*) $.
Direct computation shows that
 \be  \norm{(1-\eta) f_L + \eta  \tilde f_L}_{\cC^r(Y)} \leq   1-\eta  +
(\cO(1)/L) \eta \leq 1-\eta/2\ee
whenever  $\|f \|_{\cC^r(Y)} \leq
1$ and so when computing derivatives we obtain
\be\begin{split}   &\norm{( \cS k_q)_n(L^{-1}Y, \phi)}   \\ \leq &  \sum_{a+b=n}
\frac{n!}{a!b!}
(|q|\eta )^a (1-\eta/2)^b \norm{( k_{q})_b(Y,(1-\eta) \phi_L +
\eta \tilde \phi_L  )}\end{split}\ee
Then by summing over $n$ derivatives, 
\be
\begin{split}
& \norm{ \cS k_q(L^{-1}Y, \phi)}_h\\ 
\leq &\sum_{n=0}^{\infty}h^n\sum_{a+b=n}\frac{1}{a!b!}(|q|\eta )^a (1-\eta/2)^b \norm{( k_{q})_b(Y,(1-\eta) \phi_L +
\eta \tilde \phi_L  )}\\
\leq &\sum_{a,b=0}^{\infty}\frac{h^a}{a!}(|q|\eta )^a\frac{h^b}{b!}(1-\eta/2)^b \norm{( k_{q})_b(Y,(1-\eta) \phi_L +
\eta \tilde \phi_L  )}\\
\leq & e^{\eta h|q|}    \norm{ k_{q}(Y,(1-\eta) \phi_L +
\eta \tilde \phi_L )}_{ h(1-\eta/2)}
\end{split}
\ee
Now note that \[(1-\eta) \phi_L + \eta \tilde \phi_L =\phi_L+ \eta \phi_L(y^*)\]
and the large field regulator only contains derivatives in $\phi$, so 
\be  G(\kappa,Y,(1-\eta) \phi_L + \eta \tilde \phi_L ) 
=G(\kappa,Y,\phi_L)= G_{L^{-1}}(\kappa, L^{-1}Y, \phi) \ee
Therefore we have:
\be \norm{\cS k_q(L^{-1}Y)}_{h,G_{L^{-1}}} \leq  e^{\eta h|q|}    
  \norm{ k_{q}(Y)}_{ h(1-\eta/2),G}
 \ee
Note that this part  recovers the shrink of the analyticity region from $h+N_C$ to $h$ in the previous fluctuation step.

\item This is similar to Lemma 17 in \cite{DH00}.

Write $k_0$ as $K$ for simplicity. Then $K(X,\phi+c)=K(X,\phi)$ for any constant $c$, and 
\be
(\cS K)_n(L^{-1}Y,\phi)=(K_{L^{-1}})_n(L^{-1}Y,\phi)=K_n(Y,\phi_L)
\ee
As the estimate for the charged sector, the goal is to estimate
\be  \| (K_{L^{-1}})_n(L^{-1}Y,\phi)  \|  = \sup_{\|f_i \|_{\cC^r(L^{-1}Y)} \leq 1}
 |K_n(Y, \phi_{L};
 f_{1,L}, ... , f_{n,L})|
\ee
The supremum can be taken over fields $f_i$ such that
$f_{i,L}$ vanishes at a point in $X$. Note that for such fields, we have $\|f_{i,L}\|_{\cC^r(Y)}\le \cO(1) L^{-1}\|f_i\|_{\cC^r(L^{-1}Y)} $ (refer to Lemma 15 in \cite{DH00})
and then
\be
 \|(\cS K)_n(L^{-1}Y,\phi) \| \leq 
\| K_n(Y, \phi_{L}) \| (\cO(1) L^{-1}) ^n
  \ee
First summing only over $n \geq \dim (K)$ so we can gain
a factor $L^{-\dim(K)}$.  With  $\dim (K) =k $ we  have
\be  \sum_{n \geq k} \frac{h^n}{n!}
 \|(\cS K)_n(L^{-1}Y,\phi) \|_{} \leq  \cO(1)   L^{-k}
\| K(Y,\phi_L) \|_{{}h}   \label{high}
  \ee

For derivatives $K_n$ with $n <k$, we use the representation
\begin{align}\label{special}
&K_n(Y,\phi_L;f_{L}^{\times n})\\
=
&\sum_{m=n}^{k-1}\frac{1}{(m-n)!}  K_m(Y,0;  \phi_L^{\times (m-n)}\times
f_L^{\times n}  )\\
&+
\int^1_0\ dt
\frac{(1-t)^{k-n-1}}{(k-n-1)!}
   K_{k}(Y,t\f_{L}; \phi_{L}^{\times (k-n)}\times
f_{L}^{\times n}) 
\end{align} 
By the neutrality condition, we replace
$\phi_L$ by  $\tilde \phi_L(y)
= \phi_L(y) - \phi_L(y_*) $ for some $y_* \in Y$,  and similarly for   
$f_L$.
Then we have the estimate
\be   | K_n(Y,0;  f_L^{\times n}  ) |  \leq  (\cO(1))^n  L^{-\dim K_n}
 \| K_n(Y,0 ) \| \ \prod_{j=1}^n \|f_j\|_{\cC^r(Y)} \ee
Use this bound on the terms in the sum. The  remainder is estimated
using   $\| \tilde \phi_L \|_{\cC^r(Y)}  \leq \cO(1) L^{-1} \| \tilde \phi \|_{\cC^r(L^{-1}Y)} $.

Then we have
\begin{align}
& \| (\cS K)_n(L^{-1}Y,\f) \| \\
 \leq
&\cO(1) L^{-k}
 \Big(    \sum_{m=n}^{k-1}
  \| K_m(Y,0)\|  \| \tilde  \phi \|_{\cC^r(L^{-1}Y)}^{m-n}\\
&\qquad+\int^1_0\ dt
(1-t)^{k-n-1} \| K_k(Y,t\f_{L})\|\   \| \tilde  \phi \|_{\cC^r(L^{-1}Y)}^{k-n}
\Big) \end{align}


Now multiply by  $G(\kappa, Y, \phi)^{-1} $.  For the remainder term we use
\begin{align}  G(\kappa,Y, \phi)^{-1}
&=  G(\kappa  t^2, Y, \phi)^{-1}   G(\kappa (1- t^2), Y, \phi)^{-1}  \\ 
& \leq   G_L(\kappa  t^2, Y, \phi_L)^{-1}   G(\kappa (1- t^2), Y, \phi)^{-1}
\end{align}
Then we use
\be   \sup_\phi \ \| \tilde  \phi  \|^{a}_{\cC^r(Y)}\ G(\kappa (1- t^2), Y, 
\phi)^{-1}
 \leq   \cO(1) (\kappa(1-t^2))^{-a/2}  \ee
 This is a Sobolev inequality on derivatives of order up to
$r$ and needs   $s > r+1$.   For the zeroth derivative we  dominate
$ \tilde \phi$ by a first derivative  and  then use the Sobolev 
inequality.
Here we use the fact that $X$ is necessarily small and so has diameter  
$\cO(1)$.
This gives us
\be
\begin{split}
&\sup_{\phi}\norm{(\cS K)_n(L^{-1}Y,\phi)} G^{-1}_{L^{-1}}(\kappa,Y,\phi)\\
 \leq &\cO(1)L^{-k} \sum_{m=n}^{k-1} \sup_{\phi} \norm{K_m(Y,\phi_L)} \kappa^{-(m-n)/2} G^{-1}(\kappa,Y,\phi_L)
\end{split}
\ee
which implies the  summation over finite terms
\be
\begin{split}
&\sum_{n<k} \frac{h^n}{n!}\sup_{\phi}\norm{(\cS K)_n(L^{-1}Y,\phi)} G^{-1}_{L^{-1}}(Y,\phi)\\
\leq & \sum_{n<k} \frac{h^n}{n!} \cO(1)L^{-k} \sum_{m=n}^{k-1} \sup_{\phi} \norm{K_m(Y,\phi_L)} \kappa^{-(m-n)/2} G^{-1}(\kappa,Y,\phi_L)\\
=&  \cO(1)L^{-k} \sum_{n<k}\sum_{m=n}^{k-1}  \frac{h^n}{n!}\sup_{\phi} \norm{K_m(Y,\phi_L)}  G^{-1}(\kappa,Y,\phi_L)\\
\leq &  \cO(1)L^{-k} \norm{K(Y)}_{h,G}
\end{split}
\ee
under the the condition $k$ is no larger than 4. 
  
Combining this with (\ref{high}) we find
\be  \begin{split}
&\norm{(\cS K)(L^{-1}Y)}_{h,G_{L^{-1}}}\\
 =&\sup_{\phi} \sum_{n=0}^{\infty} \frac{h^n}{n!}\norm{(\cS K)_n(L^{-1}Y,\phi)} G^{-1}_{L^{-1}}(Y,\phi)\\
\leq  &\cO(1)  L^{-k} \norm{K(Y)}_{h,G}  +\sum_{n<k} \frac{h^n}{n!}\sup_{\phi}\norm{(\cS K)_n(L^{-1}Y,\phi)} G^{-1}_{L^{-1}}(Y,\phi)\\
\leq  &\cO(1)  L^{-k} \norm{K(Y)}_{h,G} 
  \end{split}    \ee
which implies the result.
\end{enumerate}
\end{proof}

Let \be\cF(K)=(\cS\cB K)^{\natural}\ee
By Lemma \ref{rges}, the linearization of $\cS_1\cB_1K$ in $K$ and $\sigma$ can be written as 
\be
\begin{split}
&\cS_1\cB_1 K(X,\phi,\zeta)\\
=&\sum_{\overline{Y }^L=LX }  K(Y,\phi_L+\zeta_L) +\sum_{\overline{\De }^L=LX } \Big( V(\De,\phi_L+\zeta_L)- V(\De,\phi_L)\Big)
\end{split}
\ee
where the summation is over small sets. Note that the summation $\sum_{\overline{\De }^L=LX }$  exists only if $X$ is a unit block.
Then if we combine with the Gaussian integration $\natural$,  by the the identity \[\int \zeta_L(x)\zeta_L(y)d\mu_{C_{L^{-1}}}(\zeta)=C(x-y)\] and its consequences, we have
\be \label{07052011}
\begin{split}
&(\cS_1\cB_1 K)^{\natural}(X,\phi)\\
=&
\begin{cases}
\sum_{\overline{Y }^L=LX } K^{\#}(Y,\phi_L) & \text{if $X$ is small but not unit}\\
\sum_{\overline{Y }^L=LX } K^{\#}(Y,\phi_L)+\cO(1)\sigma L^2\De C(0) &\text{if $X$ is  unit} 
\end{cases}
\end{split}
\ee
Note that here the Laplacian of covariance $\De C(0)=\cO(1)$ which is a constant independent of $\phi$.

Now  denotes $\cF_1 K$ as the linear part of $\cF$ in $K$:
\be \label{07122011b}
\cF_1K(X,\phi)=\sum_{\overline{Y }^L=LX }  K^{\#}(Y,\phi_L) 
\ee
Then for the charged and neutral sectors, we have the following estimate.

\begin{lem} \label{06162011} 
\begin{align}
\norm{\cF_1(k_q)}_{h,G(\kappa+\delta\kappa),\Gamma_p}&\leq \cO(1) e^{\eta h \abs{q}}L^2\norm{k_q}_{h(1-\eta/2),G,\Gamma_{p-r}}   \\
\norm{\cF_1(k_0)}_{h,G(\kappa+\delta\kappa),\Gamma_p}&\leq \cO(1)L^{2-\dim(k_0)}\norm{k_0}_{h,G,\Gamma_{p-r}}
\end{align}
for $0\leq \eta\leq 1$ and any $p,r\geq 0$.
\end{lem}
\begin{proof}
By \ref{07122011b} and Lemma \ref{rges},
\be \label{07122011c}
\begin{split}
\norm{\cF_1K(X)}_{h,G(\kappa+\delta\kappa)} 
\leq &\sum_{\overline{Y}^L=LX } \norm{K^{\#}(Y)}_{h,G_L(\kappa+\delta\kappa)} \\
\leq &\sum_{\overline{Y}^L=LX } \norm{K(Y}_{h,\hat{G}_{}(\kappa)} 2^{\abs{Y}}\\
\leq &\sum_{\overline{Y}^L=LX } \norm{\cS K(L^{-1}Y)}_{h,G_{L^{-1}}(\kappa)} 2^{\abs{Y}}\\
\end{split}
\ee
by the property $\hat{G}\geq G$.
Then by Lemma \ref{e0051} for the charged sectors ,  we have
\be
\begin{split}
\norm{\cF_1k_q(X)}_{h,G(\kappa+\delta\kappa)} 
\leq & \sum_{\overline{Y}^L=LX }\cO(1) e^{\eta h\abs{q}}\norm{k_q(Y)}_{h(1-\eta/2),G_{}}2^{\abs{Y}}
\end{split}
\ee

By combining the  property \ref{07132011} of the large set regulator $\Gamma_p$ ,  we have
\[\norm{\cF_1(k_q)}_{h,G(\kappa+\delta\kappa),\Gamma_p}\leq \cO(1) e^{\eta h \abs{q}}L^2\norm{k_q}_{h(1-\eta/2),G,\Gamma_{p-r}}  \]
where the $L^2$ factor is from the summation over $LX$.

Similarly for the neutral sectors.
\end{proof}
Note that the formula \ref{07052011} holds for any polymer activity $K$, not only for $k_q$ or $k_0$. And similar to \ref{07122011c}, we have
\be \label{07122011d}
\norm{\cF_1(K)}_{h,G(\kappa+\delta\kappa),\Gamma_p}\leq \cO(1)L^{2}\norm{K}_{h,G,\Gamma_{p-r}}
\ee

\subsubsection{Summary}   \label{partitionsummary}
Now we  specify the norms in different scales . As in \cite{DH00}, we shall use a family of polymer activity norms defined for  the corresponding scales $j=0,1,2,...$ 
by
\be \norm{K}_j=\norm{K}_{G(\kappa_j),h_j, \Gamma}  \ee

And the large field regulator  $G(\kappa_j)$ is associated with
\be \kappa_j=\kappa_0 \left(\sum^{j}_{k=0} 2^{-k}\right) \ee
Since $\kappa_j$ increases slowly in $j$,
therefore the domain of analyticity which is defined by
\be h_j=h_{\infty} \left(1+\sum^{\infty}_{k=j+1} 2^{-k}\right) \ee
with  $h_{\infty}=\kappa_0^{-1/2}$ (so  $h_{\infty} \geq \cO (L^{s/2} )$),
will  also decrease slowly in $j$, but not shrink to 0 eventually.

Note that those constants satisfy $ \cO(1) h_j^2\leq 1/\kappa_j\leq \cO(1) h_j^2$ for any $j$.

Here we summarize the requirement of our renormalization analysis. 

\begin{thm} \label{RGcontrol}
\mbox{} 
There exists $\epsilon>0$ so that
if  $\sigma_j<2^{-j}\epsilon$ and $\norm{K_j}<2^{-j}\epsilon$ , then one can choose extraction so
\begin{align}
\sigma_{j+1}&=\sigma_j+\alpha_j(K_j)\\
K_{j+1}&=\cL_j K_j+g_j(\sigma_j,K_j)
\end{align}
where $\alpha_j$ and $\cL$ are linear, and $g_j$ is smooth with derivatives bounded uniformly in $j$ and both $g(0,0)=0$ and $Dg(0,0)=0$.

If $\dim \overline{K}_j \geq 4$, then
\begin{align}
\abs{\alpha_j(K_j)}&\leq \cO(1)\norm{K_j}_j\\
\norm{\cL_j K_j}_{j+1}&\leq \delta\norm{K_j}_j 
\end{align}
with $\delta=\cO(1)\max({L^{2-\beta/4\pi},L^{-2}})$, and $\dim \overline{K}_{j+1} \geq 4$.
\end{thm}
\begin{remark} Note that Theorem \ref{RGcontrol} cannot be iterated until further modification because of the $j-$dependent condition on $\sigma_j$. 

\end{remark}
\begin{proof} 
The standard forms of
\begin{align}
\sigma_{j+1}&=\sigma_j+\alpha_j(K_j)
\end{align}
are constructed in Lemma \ref{ve000}.

The bound on $\alpha_j$ is straightforward by Lemma \ref{ve000}.

The essential part is to  control the growth of $K_j\rightarrow K_{j+1}$.
To get the precise estimate of $K_{j+1}\rightarrow K_j$, we follow the method in \cite{DH00}. For simplicity, write 
\[K_{j+1}=\cE((\cS\cB K_j)^{\natural},F_j)=\cR (K_j,F_j)\]
 as the corresponding Renormalization Group transformation, where $F_j$ is the extraction. Then  under the assumptions as Lemma \ref{rges},   we make the decomposition as:
\be\label{4terms}
K_{j+1}=\cR_{\geq 2}(K_j,\sigma_j) +  \cE_1\cF_1(K_j)+\cE_1(\cO(1)L^2\sigma_j\De C(0))
\ee
where $\cR_{\geq 2}(K_j,\sigma_j)$ means the higher order terms in $(K_j,\sigma_j)$.

Now note that the extraction is made after reblocking, scaling and Gaussian integration, so the polymer activity $K$ in the extraction formula \ref{06302011b} should be replaced by $(\cS_1\cB_1 K)^{\natural}$ defined in \ref{07052011}. Then the extraction formula will be
\be \label{07142011d}
\begin{cases}
V_{F_0}(\De)&=\sum\limits_{X\supset\De} F_0(X,\De) = \de E+\cO(1)L^2\sigma_j\De C(0)\\
V_{F_1}(\De)&=\sum\limits_{X\supset\De} F_1(X,\De)=-\frac{\de\sigma}{2\beta}\int_{x\in\De} (\partial\phi)^2   
\end{cases} 
\ee
In other words $\cE_1(\cO(1)L^2\sigma_j\De C(0))=0$  since it is $\phi-$independent and canceled in the extraction step.

For the linear term $\cE_1\cF_1(K_j)$, we decompose it as
\be \label{07172011}
 \cE_1\cF_1(K_j)=
\cE_1\cF_1( K_j \cdot\chi_{\bar \cS})+ \cE_1\cF_1\left(\sum_{q\ne 0}  k_q\cdot\chi_\cS\right)  +  \cE_1\cF_1(k_0 \cdot\chi_\cS)
\ee
where $\chi_\cS$ means the characteristic function on small sets and $\chi_{\bar \cS}$ for large sets.
We shall treat those  terms independently. From now on, the term $\de f$ in Lemma \ref{rges} will be bounded by $K$ by the arguement in Section \ref{VEstimates}.
\begin{enumerate}

\item Estimates on the large set terms $\cE_1\cF_1( K_j \cdot\chi_{\bar \cS})$:

Since there is no extraction made from large sets and there is no contribution from $\sigma$ by \ref{07052011}, by the built-in regulator property 
\be \Gamma_p(L^{-1}\overline{X}^L)\leq L^{-4}\Gamma_p(X) \qquad \mbox{for large sets $X$} \ee we have:
\be \label{07172011b}
\begin{split}
\norm{\cE_1\cF_1( K_j \cdot\chi_{\bar \cS})}_{j+1}&=\norm{\cE_1\cF_1\left(K_j \cdot \chi_{\bar{\cS}}      \right)}_{j+1}\\
&\leq \cO(1) L^{-2}\norm{K_j}_j \\
&\leq  \frac{\de} {4}\norm{K_j}_j
\end{split}
\ee

\item Estimates on the charged terms $\cE_1\cF_1(\sum_{q\neq 0}  k_q\cdot\chi_\cS)$:

Since there is no extraction made from charged terms, by choosing $\eta=2h^{-1}N_C\leq 1$ in Lemma \ref{06162011}
\be\begin{split}
&\norm{\cE_1\cF_1(\sum_{q\ne 0}  k_q\cdot\chi_\cS)}_{j+1} \\
\leq &\cO(1)L^2 \sum_{q\neq 0}e^{ -|q|(\beta C(0)-2N_{\beta C})  +\beta C(0)/2}\norm{K_j}_j  
\end{split}
\ee
Here we have used the fact that $\norm{k_q}_j\leq\norm{K}_j$ according to the definition \ref{06282011}.
Recall that 
\be
C(0)=\int_1^L \frac{1}{l}u(0)dl=\frac{1}{2\pi}\log L
\ee
Therefore for $\beta>8\pi$
\be \begin{split}
\norm{\cE_1\cF_1\Big(\sum_{q\ne 0}  k_q\cdot\chi_\cS\Big)}_{j+1} 
\leq &   \cO(1) L^{2-\beta/4\pi}\norm{K_j}_j\\
\leq  &\frac{\de}{4}\norm{K_j}_j
\end{split}\ee

\item Estimates on the neutral terms $\cE_1\cF_1(k_0 \cdot\chi_\cS)$:


Recall that the extraction is made from the  neutral terms. More precisely, the $\sigma-$independent part of the extraction is made from $\cF_1(k_0\chi_S)=(\cS_1\cB_1(k_0\chi_S))^{\natural}$. 
Since $k_0=\overline{K_j}$, by $\dim \overline{K}_j\geq 4$ and Lemma \ref{06162011} we have 
\be \begin{split}
\norm{\cF_1 k_0}_{h,G(\kappa+\de\kappa),\Gamma_p} &\leq \cO(1) L^{2-\dim k_0}\norm{k_0}_{h,G,\Gamma_{p-r}}\\
&\leq \cO(1) L^{-2}\norm{k_0}_{h,G,\Gamma_{p-r}}
\end{split}\ee
Then
\be \label{06302011}
\begin{split}
&\norm{\cE_1\cF_1(k_0 \cdot\chi_\cS)}_{h,G(\kappa+\de\kappa),\Gamma_p}\\
\leq & \norm{\cF_1(k_0 \cdot\chi_\cS)-F(\cF_1(k_0 \cdot\chi_\cS))}_{h,G(\kappa+\de\kappa),\Gamma_{p}}\\
\leq &\cO(1) L^{-2}\norm{k_0}_{h,G,\Gamma_{p}} + \norm{F(\cF_1(k_0 \cdot\chi_\cS))}_{h,G(\kappa+\de\kappa),\Gamma_{p}}
\end{split} \ee
Note that the extraction $F(\cF_1(k_0 \cdot\chi_\cS))$ is made from $\cF_1(k_0 \cdot\chi_\cS)$, which implies
\be \begin{split} \label{07142011}
 \norm{F(\cF_1(k_0 \cdot\chi_\cS))}_{h,G(\kappa+\de\kappa),\Gamma_{p+2}}
  \leq &\cO(1)\norm{\cF_1(k_0 \cdot\chi_\cS)}_{h,G(\kappa+\de\kappa),\Gamma_{p+2}}\\
   \leq &\cO(1) L^{-2}\norm{k_0}_{h,G,\Gamma_{p}}
\end{split}\ee
Substitute this bound into \ref{06302011} and combine with the bound $\norm{k_0}_j\leq \norm{K_j}_j$, we have
\be
\norm{\cE_1\cF_1(k_0 \cdot\chi_\cS)}_{j+1}\leq   \cO(1)L^{-2}\norm{K_j}_j \leq    \frac{\de}{4}\norm{K_j}_j   
\ee
Also note that here the hypothesis $\dim \overline{K}_j\geq 4$   implies 
\[\dim \overline{K}_{j+1}=\dim (\cF\overline{K}_{j}-F(\overline{K}_{j}))\geq 4\] 
by the construction of the extraction \ref{06302011} for any $j$. And it is true for any $j$ by induction.
\end{enumerate}

%

In summary,  for $K_{j+1}=K_{j+1}(K_j,\sigma_{j})$, say the linearization  of $K_{j+1}$ at $(0,0)$ with respect to $K_j$ as $\cE_1\cF_1 K_j=\cL_jK_j$, then
\be
K_{j+1}=\cL_jK_j+g_j(\sigma_j,K_j)
\ee
This linear operation here is from the last three terms of \ref{4terms}. And the bounds follows immediately by \ref{kfform}.
\be
\norm{\cL_j K_j}_{j+1}= \norm{\cE_1\cF_1 K_j}_{j+1} \leq \delta\norm{K_j}_j \label{kfform}
\ee

Also $g_j(\sigma_j,K_j)=\cR_{\geq 2}(K_j,\sigma_j)$, which is the higher order terms. The condition $g_j(0,0)=Dg_j(0,0)=0$ follows from the construction of section \ref{fundmentaltransform}.
We use a Cauchy bound to show the uniform boundedness in $j$. Note that $\cR (sK,sF)$ is well defined for all complex $s$ in the disc $\abs{s}\leq D$, where $D$ is a constant to be chosen. Then the bound
\[
\norm{\cR(sK,sF)}_{j+1}\leq\cO(1)DL^2 (\norm{K}_j+\abs{\sigma_j/\de\kappa_j})
\]
follows directly from the estimate \ref{nrges}. Here we bounded the $\de f$ term by $\norm{K}$ by \ref{07122011}. Since 
\[\sigma_j<2^{-j}\epsilon\qquad\mbox{and}\qquad \de\kappa_j=2^{-j-1}\kappa_0\]
so 
\be
\norm{\cR(sK_j,sF_j)}_{j+1}\leq\cO(1)DL^2 (\norm{K_j}_j+\abs{\epsilon/\kappa_0})
\ee
Also by the analyticity of $\cR(sK_j,sF_j)$ in $s$, we apply the residue theorem as
\be\cR_{\ge 2}(K_j,\sigma_j)=\frac{1}{2\pi i}\int_{\abs{s}=D}\frac{\cR(sK_j,sF_j)}{s^2(s-1)} ds \label{residue00}
\ee
This is true for any $D\geq 2$. By applying the norms to \ref{residue00}, we will have:
\be \begin{split}\norm{\cR_{\ge 2}(K_j)}_{j+1} \leq& \cO(1)D^{-2}\sup_{\abs{s}=D}\norm{\cR(sK_j,sF_j)}_{j+1}\\
\leq &\cO(1)D^{-1}L^2 (\norm{K_j}_j+\abs{\epsilon/\kappa_0})
\end{split}\ee
This bound is independent of $j$ by choosing $D=2$ for example.

And since $g(K_j,\sigma_j)$ is analytic in $K_j$ and $\sigma_j$, so the derivatives of $g$ are also bounded uniformly in $j$ by a Cauchy bound.

This completes the proof of Theorem \ref{RGcontrol}.

\end{proof}
\subsection{Tuning Program} \label{tuning}
So far we have established the R.G. flow $(V_j,K_j)\rightarrow(V_{j+1},K_{j+1})$ in the actual form of $(\sigma_j,K_j)\rightarrow(\sigma_{j+1},K_{j+1})$. The $\cE_j$ terms are usually interpreted as the energy densities and $\sigma_j$ terms as the field strength in QFT. The flow apparently has a trivial fixed point $(0,0)$. Therefore the scaling limit problem  as $M\rightarrow\infty$ arises as the tuning program, which drives the flow to the fixed point.
Here we quote a version of Stable Manifold Theorem in \cite{B09}. Another similar usage of this approach can be found in  \cite{D09}.

For any Banach space $X$, $B_{X,r}$ denotes the open ball of radius $r$ centered on the origin. For $j\in\bN_0$, let $E_j,F_j$ be Banach spaces. Let $B_{E_j,r}\subset E_j$ and $B_{F_j,r}\subset F_j$ be balls of radius $r$ centered at the origin. Suppose for each $j\in\bN$ we have a map from $B_{E_{j-1},r}\times B_{F_{j-1},r}$ to $E_j\times F_j$ given by
\begin{align}
x_j&=A_jx_{j-1}+B_jy_{j-1}+f_j(x_{j-1},y_{j-1})\\
y_j&=C_jy_{j-1}+g_j(x_{j-1},y_{j-1})
\end{align}
where $A_j,B_j,C_j$ are linear and $f_j,g_j$ are smooth functions satisfying $f_j(0,0)=0=g_j(0,0)$, $Df_j(0,0)=0=Dg_j(0,0)$.
\begin{lem}(Stable Manifold Theorem) \label{SMT}
For $j\in\bN$, let $f_j,g_j$ be smooth functions uniformly in $j$, let $A_j$ be invertible, $\sup\limits_{j,k}\norm{A_j^{-1}}\norm{C_k}<1$, $\sup\limits_j\norm{C_j}<1$ and $\sup\limits_j\norm{B_j}<\infty$. Then there exists a ball $B_{F_0,\rho}$ and a smooth function $h:B_{F_0,\rho}\rightarrow B_{E_0}$ such that if $(x_0,y_0)$ lies in the graph $\{(h(y),y):y\in B_{F_0,\rho}\}$, then $(x_j,y_j)\rightarrow 0$.
\end{lem}
The proof of the stable manifold theorem can be found in \cite{B09}. And we think there is no other short way to claim the existence of the R.G. flow except this lemma. Here the point is  the smoothness of the maps and the completeness of the Banach space play important roles in the proof the stable manifold theorem.

By the remark in \cite{B09}, if we define a Banach space $\cZ$ with the norm
\be
\norm{u}=\sup_j \mu^{-j} \max\{\norm{x_j},\norm{y_j}\}
\ee
then by the proof of Theorem \cite{B09}, the parameter $\mu$ should satisfy
\be   \mu\in \left(\norm{C_j},\norm{A^{-1}_j}^{-1}\right)\cap(0,1)\qquad\text{for any $j$} \label{07082011}\ee
Also it is easy to see that \[\norm{x_j}\leq \mu^j \norm{u}\qquad \text{and}\qquad\norm{y_j}\leq \mu^j \norm{u}\]

Now let us specify the Banach space on which we need $K_{j+1}$ is smooth. Note that $K_{j+1}$ can be regarded as
\be
K_{j+1}=K_{j+1}(\sigma_j,K_{j})
\ee
is a map from a ball $\bR  \times \cB_{j}$ as Section 6.1 of \cite{B09}. Here each $\cB_j$ is a Banach space with the norm on the corresponding polymer activities. The completeness of the space under this norm is also very crucial in the proof of the stable manifold theorem. And we assume the domain of $K_{j+1}$ is a ball in $\bR  \times \cB_{j}$ with radius $r$ defined by
\be
\abs{\sigma_j}\leq r\qquad  \norm{K_j}\leq r   \label{062820113}
\ee


\subsection{Main Result on Partition Function}

Finally, based on the previous analysis, we can conclude as the following result, which is also the restatement of Theorem \ref{pthm}: 
\begin{thm} \label{pproof}
For sufficiently small $z$ and $\sigma_0=\sigma(z)$, and $0<\epsilon<1$,
\be
Z(\La_M,z,\sigma(z))=e^{\cE_j}\int \cE xp(\Box e^{-V_j} +K_j)(\La_{M-j},\phi)  d\mu_{\beta v_{M-j}}(\phi)
\ee
hold for any $0\leq j\leq M-1$, where
\begin{align}
V_j(\Delta)&=\sigma_j\int_{\Delta} (\partial\phi)^2\\
\cE_j&=\sum_{k=0}^{j-1}\de\cE_k\\
\sigma_j&=\sigma_0-\sum_{k=1}^{j}\de\sigma_k
\end{align}
and with such $\sigma_0$, the flow $(\sigma_j,K_j)\rightarrow (0,0)$. And $\dim \overline{K}_j \geq 4$.

More precisely, they satisfy the bounds: 
\begin{align}
\abs{\delta \cE_j}&\leq \tilde{C}(L)\de^j\epsilon \abs{\La_{M-j}} \\
\abs{ \sigma_j}&\leq  \de^j\epsilon  \\
\norm{K_j}_j&\leq \de^j \epsilon \label{e1}
\end{align}
where $\delta=\cO(1)\max({L^{2-\beta/4\pi},L^{-2}})$, $\cO(1)$ is a constant independent of $L$ and $j$, and $\tilde{C}(L)$ is a constant which  only depends on $L$.

\end{thm}

\begin{proof} \mbox{}

First we need to show the existence of the tuning coupling constant by verifying the hypothesis of the stable manifold theorem Lemma \ref{SMT}. The argument is analogous to dipole gas model \cite{D09}. However here the hypothesis in Theorem \ref{RGcontrol}, such as $\sigma_j<2^{-j}\epsilon$, does not fit in the Stable Manifold Theorem \ref{SMT} very well. To show the existence, further modification is needed.

So far we have established the estimate $(V_j,K_j)\rightarrow (V_{j+1},K_{j+1})$, which may also be regarded as $(\sigma_j,K_j)\rightarrow (\sigma_{j+1},K_{j+1})$. Now we write 
\be\tilde{\sigma_j}=2^j\sigma_j\qquad\mbox{and}\qquad\tilde{K_j}=2^jK_j \label{05022011}\ee
Then this implies the radius $r=\epsilon$ of the ball in \ref{062820113}  of  the Theorem \ref{SMT}:
 \be   \abs{\tilde{\sigma_j}}\leq \epsilon     \qquad\mbox{and}\qquad \norm{\tilde{K_j}}_j\leq \epsilon \ee
By our construction as Theorem \ref{RGcontrol}, we have
\be
\tilde{\sigma}_{j+1}=2\tilde{\sigma_j}+2\alpha_j(\tilde{K_j})
\ee
so $A_j=2, B_j=\alpha_j$ and $f_j(0,0)=Df_j(0,0)=0$ in Lemma \ref{SMT}.

And for
\be
\tilde{K}_{j+1}=2\cL_j\tilde{K_j}+2^{j+1}g_j(2^{-j}\tilde{\sigma_j},2^{-j}\tilde{K_j})
\ee
i.e. $C_j=2\cL_j$ in Lemma \ref{SMT}.
Since by our previous estimates  $\cL_j$ is contractive with the factor $\delta$, so if  $L$ large, then $\delta$ is small,  and we have
\be
\sup_{j,k}\norm{A_j^{-1}}\cdot\norm{C_k}<1 \qquad \mbox{and}\qquad\sup_j\norm{C_j}<1
\ee
And the higher orders of $g_j$ satisfy $g_j(0,0)=Dg_j(0,0)=0$.

Let \[\tilde{\alpha}_j(\tilde{K}_j)=2\alpha_j(\tilde{K}_j)\qquad \mbox{and}\qquad \tilde{g_j}(\tilde{\sigma}_j,\tilde{K}_j)=2^{j+1}g_j(2^{-j}\tilde{\sigma}_j,2^{-j}\tilde{K}_j)\] 
First we need to check the smoothness and uniform bounds of new $\tilde{f}_j$ and $\tilde{g}_j's$ in $j$. And for $\tilde{g}_j$, the smoothness condition follows from Theorem \ref{RGcontrol} similarly. And for the uniform boundedness, we consider
\be
\begin{split}
&\tilde{g}_j(\tilde{\sigma}_j,\tilde{K}_j)\\
=&\tilde{g}_j(\tilde{\sigma}_j,\tilde{K}_j)-\tilde{g}_j(0,0)\\
=&\int_0^1 \frac{d}{dt}\tilde{g}_j(t\tilde{\sigma}_j,t\tilde{K}_j) dt\\
=&\int_0^1 \left( \tilde{\sigma}_j D_{\tilde{\sigma_j}} \tilde{g}_j   +  < \tilde{K}_j,  D_{\tilde{K_j}} \tilde{g}_j  >                                     \right) dt\\
=&\int_0^1 \left( 2\tilde{\sigma}_j D_{\sigma_j} g_j (2^{-j}\tilde{\sigma}_j,2^{-j}\tilde{K}_j)  +   2<\tilde{K}_j , D_{K_j} g_j                                    (2^{-j}\tilde{\sigma}_j,2^{-j}\tilde{K}_j) >  \right) dt
\end{split} \ee
The last step is by chain rule of the Frechet derivatives. Then by applying norms onto the above equation, the uniform bound follows similarly as Theorem \ref{RGcontrol}. And the derivatives argument is also similar. 


This completes the verification of the stable manifold theorem, so there exists $\tilde{\sigma}_0>0$ such that the flow $(\tilde{\sigma}_j,\tilde{K}_j)\rightarrow 0$. Therefore immediately we have the existence of the flow $(\sigma_j,K_j)$ as a consequence.

The forms  of $V_j$  follow from Lemma \ref{ve000}. And by the construction $\sigma_{j+1}=\sigma_j-\delta\sigma_j$.\,  we are removing a small quantity from $\sigma_0=\sigma$ each time during the renormalization group transformations.

 We prove the estimates on $K_j$ by induction. For the initial cluster expansion, by Lemma \ref{k00estimate}, $\norm{K_0}_0\leq \cO(1)\delta^0\abs{z}^{1-\epsilon}\leq \cO(1)\delta^0\epsilon$ for the activity $\abs{z}\leq \cO(1)\epsilon$.
Note that for $L$ large, the choice such as $\mu=3\sup_j\norm{\cL_j}=\de$ fits the condition \ref{07082011}.  Therefore by the stable manifold theorem, we have for all $j$
\be
\norm{K_j}_j\leq \de ^j \norm{u}\leq \de ^j\epsilon \quad\text{and}\quad \abs{\sigma_j}\leq \de ^j \norm{u} \leq\de ^j \epsilon
\ee 
Then by Lemma \ref{ve000},  \ref{07052011} and \ref{07122011d}, we have
\be
\begin{split}
\abs{\delta \cE_j}
&=\delta E_j\abs{\Lambda_{M-j-1}} \leq \norm{(\cS_1\cB_1 K_j)^{\natural}}_j\abs{\Lambda_{M-j-1}} \\
&  \leq \cO(1)(L^2\norm{ K_j}_j+\cO(1)AL^2\abs{\sigma_j})\abs{\Lambda_{M-j-1}}        \leq \tilde{C}(L)\de^j\epsilon\abs{\La_{M-j}}
\end{split}
\ee
This completes the proof of Theorem \ref{pproof}.
\end{proof}

\section{Extension to Correlations}

This section is the extension of the renormalization group analysis in the previous sections onto the correlation problem.

The generating functional for charge correlations is defined as:
\be  
Z(\rho)=<e^{i(\partial\phi,\rho)}>=Z(0)^{-1}\int e^{i(\partial\phi,\rho)} e^{zW-\sigma V} d\mu_{\beta v_M}(\phi)
\ee
where $(\phi,\rho)=\int \phi(x)\rho(x)dx$ and $Z(0)$ is the partition without any external charges.

We consider particularly \be(\rho,\partial\phi)=\la_1\bn_1\cdot\partial\phi(a)+\la_2\bn_2\cdot\partial\phi(b)\ee for $a,b\in\La_M$ , $\abs{a-b}\gg1$ and $\bn_1$ and $\bn_2$ are any unit vectors. Then $\int \rho=0$. 
Note that here we are anticipating the fact that at the low temperature the external charges will form into dipoles  with the nearest opposite charges, therefore the natural consideration of correlations shall be the dipole-dipole correlation, or $\partial\phi$ correlation analogue to dipole gas correlations. And we will need the analyticity in $\la$ on a small ball with the radius $\lambda_0$ ,i.e. $\abs{\la_1}\leq \la_0$ and $\abs{\la_2}\leq \la_0$. 

Then the truncated dipole$-$dipole field correlation function will be
\be  
\left<\partial\phi(a)\partial\phi(b)\right>^T=(-1)\cdot \left[\frac{\partial^2}{\partial \la_1\partial\la_2}\log Z(\rho) \right] \Big|_{\rho=0}
\ee

\subsection{R.G. Analysis}

\subsubsection{Initial Cluster Expansion}
As the vacuum case, the initial cluster expansion is the basic set up for the iterating of renormalization group transformations. Now for the non-vacuum case, we still write the interaction as a sum over the unit blocks, make a Mayer expansion, then group together into connected polymers.
And we will use $2\pi$ translation invariance as the previous argument. That is the advantage of $\partial\phi$ correlation instead of $\phi$.
\begin{align}
&Z'(\rho)\\=&\int \exp(i(\partial\phi,\rho)+zW(\La,\phi)-\sigma V(\La,\phi))d\mu_{\beta v_M}(\phi)\\
=&\int \exp(i(\la_1\bn_1\cdot\partial\phi(a)+\la_2\bn_2\cdot\partial\phi(b))+zW(\La,\phi)\sigma V(\La,\phi))d\mu_{\beta v_M}(\phi)
\end{align}

Let
\be
(\partial \phi,\rho)=\sum_{\De\subset\La_M}(\partial \phi,\rho_{\De})
\ee
where 
\be
(\partial \phi,\rho_{\De})=\la_1\bn_1\cdot\partial\phi(a)\chi_{\De\ni a}+\la_2\bn_2\cdot\partial\phi(b)\chi_{\De\ni b}
\ee

Now we write the interaction term as:
\be
\begin{split}
&\exp \left(i(\partial\phi,\rho)+zW(\La_M)-\sigma V(\La_M)\right) \\
=&\prod_{\Delta\subset\La_M} \exp \left( i(\partial \phi,\rho_{\De})+zW(\Delta)-\sigma V(\Delta)\right)\\
=&\prod_{\Delta\subset\La_M} \left[  \left(e^{zW(\Delta)+i(\partial \phi,\rho_{\De})}-1 \right)e^{-\sigma V(\Delta)}   +e^{-\sigma V(\Delta)}  \right] \\
=&\sum_{X\subset \La_M}   \prod_{\Delta\in X}  \left(e^{zW(\Delta)+i(\partial \phi,\rho_{\De})}-1 \right)e^{-\sigma V(\Delta)}   \prod_{\Delta\notin X}   e^{-\sigma V(\Delta)}\\
=&\cE xp\left(\Box e^{-V_0(\phi)}+K_0(\phi,\rho)\right)(\La_M)
\end{split} \ee
where 
\be
\begin{split}
K_0(X,\rho)=&\prod_{\Delta\subset X}K_0(\Delta,\phi,\rho)\\
=&\prod_{\Delta\subset X}(\exp(i(\partial \phi,\rho_{\De})+zW(\Delta,\phi))-1)e^{-\sigma V(\Delta,\phi)}
\end{split}
\ee
Note that the non-vacuum activities are $\rho-$ dependent. Furthermore, they have \textit{pinning property} , i.e, $K_0(X,\rho)=K_0(X,\rho=0)$ unless the polymer $X$ contains one or both external charges $a$ and $b$. Later the pinning property will guarantee the convergence of the series.  First of  all, we need an estimate on $K_0$. Note as the vacuum case, the estimate on $V_0$ is trivial.

\begin{lem}Estimate on norm of $K_0$
\label{07162011c}

For $\abs{z}$ sufficiently small and $0<\epsilon<1$,
\begin{align}
\norm{K_0(X)}_{h,G} \leq
\begin{cases}
(\abs{z}^{1-\epsilon})^{\abs{X}} & \text{if $X$ contains no external charges;}\\
(\abs{z}^{1-\epsilon})^{\max(\abs{X}-m,0)}C_{\la} & \text{if $X$ contains $m$ external charges.}
\end{cases}
\end{align}
where  $C_{\la}$ is a  constant
 \be C_{\la}= \la_0\exp(h+\abs{z}\la_0^{-1}e^{2h})\ee
 with $\la_0$ small.
\end{lem}
\begin{proof}

\mbox{}

If both $a\notin \Delta$ and $b\notin \Delta$, then by the pinning property $K_0(\Delta,\rho)=K_0(\Delta,\rho=0)$, we have
\be
\norm{K_0(\Delta,\rho)}_{h,G}\leq \cO(1)\abs{z}^{1-\epsilon}
\ee

If  $a\in \Delta$, then by the assumption that $a$ and $b$ are far apart, $b\notin \Delta$ since $\Delta$ is the unit block. Then say
\be
K_0(\Delta,\rho)=
\exp(i(\la_1\bn_1\cdot\partial\phi(a))+zW(\Delta,\phi))-1
\ee
Let $\tilde{W}(\rho)=\tilde{W}(\Delta,\rho,z,\phi):=i\la_1\bn_1\cdot\partial\phi(a)+zW(\Delta,\phi)$, then $K_0=e^{\tilde{W}(\rho)}-1$. Note that the estimate on $e^{\sigma V}$ term is given by Lemma \ref{k0estimate}. We shall estimate the norm of $K_0$ following by the similar steps as the vacuum case. 
\begin{align}
&\tilde{W}_n(\Delta,\phi,\rho;f_1,\dots,f_n)\\
=& \frac{\de^n}{\de t_1\dots\de t_n}  \tilde{W}(X,\phi+t_1f_1+\dots+t_nf_n,\rho) \Big| _{t=0}\\
=& zW_n(\Delta,\phi;f_1,\dots,f_n)\\
&\qquad\qquad+\frac{\de^n}{\de t_1\dots\de t_n}  i\la_1\bn_1\cdot(\partial\phi(a)+t_1\partial f_1(a)+\dots+t_n\partial f_n(a)) \Big| _{t=0}\\
=&
\begin{cases}
zW_n(\Delta,\phi;f_1,\dots,f_n)+i\la_1\bn_1\cdot\partial f_1(a) & \text{if $n=1$}\\
zW_n(\Delta,\phi;f_1,\dots,f_n) &\text{if $n \geq 2$} 
\end{cases}
\end{align}
which implies
\begin{align}
\norm{\tilde{W}_n(\Delta,\phi,\rho)}
 =&\sup_{f_i \in H(X)\atop \|f_i\|_{\infty,r,X} \leq 1} \abs{\tilde{W}_n(\Delta,\phi,\rho;f_1,...,f_n)}\\
\leq & \abs{z}\norm{W_n(\Delta,\phi)}+\abs{\la_1}\\
\leq & \abs{z}+\abs{\la_1}
\end{align}
The last step is by the direct computation that $\norm{W_n(\Delta,\phi)}\leq 1$.
Furthermore, the $n-th$ derivative is bounded by
 \begin{align}
  \frac{(2h)^n}{n!}\norm{\left(e^{\tilde{W}(\De)}\right)_n (\phi)}
 &\leq \exp\left[ \sum_{n=0}^{\infty}  \norm{\tilde{W}_n(X,\phi,\rho)}              \right]\\
 &\leq \exp\left[\abs{\la_1}h+\sum_{n=0}^{\infty} \frac{(2h)^n}{n!} |z| \right]
 \end{align}
  so by taking the supremum over $\phi$, multiplying by $2^{-n}$ and summing over $n$, we have
 \be
 \norm{e^{\tilde{W}(\De)}}_{h,G=1}\leq2\exp\left[\abs{\la_1}h+ |z|e^{2h}  \right] \ee
 For the last step, we use the Residue Formula:
 \be
 e^{\tilde{W}(\De)}-1=\frac{1}{2\pi i}\int_{|\gamma|=\la_0^{-1}}\frac{e^{\gamma \tilde{W}(\De)}}{\gamma(\gamma-1)} d\gamma
 \ee
Note that the choice of the contour $|\gamma|=\la_0^{-1}$ here is different from the vacuum case, and should be $z-$independent otherwise it will be out of control because of  the external $\rho$.  Therefore, if $\la_0$ is sufficiently small, then
 \be
 \begin{split}
 \norm{e^{\tilde{W}(\De)}-1}_{h,G=1}&\leq\frac{1}{\pi} \la_0\norm{e^{\gamma \tilde{W}(\De)}}_{h,G=1}\leq \frac{2\la_0}{\pi}\exp[|\gamma|\cdot(|\la_1|h+|z|e^{2h})]\\
 & \leq \la_0\exp(h+\abs{z}\la_0^{-1}e^{2h})
 \end{split}
 \ee
Finally, by the factorization property of $K_0$, we complete the proof of the conclusion.
\end{proof}
\begin{remark}
\mbox{}
\begin{enumerate}
	\item Note for these non-vacuum polymer activities, they are not necessarily  small quantities even if the activity $\abs{z}$ is sufficiently small. For example, if for a unit block which contains the external charge $a$, $\abs{X}=m=1$, then $\norm{K_0(X)}$ will be constant quantities. But when the radius of analyticity $\la_0$ is small, $K_0$ will be sufficiently small.
	\item  If there are infinitely many external charges,  our argument will fail since it will lead us to a divergent series. However, if there are only finitely many different terms from vacuum case, the pinning property will provide the necessary conditions, since the difference between them is just a few unit blocks by the factorization property of norms on polymer activities.

\end{enumerate}

\end{remark}

\subsubsection{$\rho-$dependent Extraction} \label{05252011}
Now to continue the following renormalization group transformations, we need to construct the $\rho-$dependent polymer activities $V_j$ and $K_j$ in each step, which are mainly deduced by the extraction we made in each step. 
As the vacuum case, the extraction is still made from the neutral sectors 
\be\bar K(X,\phi)=\frac{1}{2\pi}\int_{-\pi}^{\pi} K(X,\Phi+\phi)d\Phi\ee 
on small sets, and composed of two pieces: relevant and marginal $K(X,\phi,\rho=0)$ and $\rho-$dependent $K(X,\phi,\rho)$.

Let \be F(X)=\sum_{\De\subset X} F(X,\De,\phi,\rho)\ee and with the corresponding $F=F_0+F_1$,
\begin{align}
	F(X,\De,\phi,\rho)=    \alpha^{(0)}(X,\rho)
	+ & \sum_{\mu,\nu} \alpha_{\mu,\nu}^{(2)}(X,\rho=0)  \int_{\De}(\partial_{\mu}(\phi))(\partial_{\nu}(\phi))\\
+&\sum_{\mu,\nu\rho}\alpha_{\mu,\nu\rho}^{(2)}(X,\rho=0)\int_{\De}(\partial_{\mu}(\phi))(\partial_{\nu\rho}^2(\phi))
\end{align}
Here we determine those coefficients to be:
\be   
\begin{cases}
\alpha^{(0)}(X,\rho)=\frac{1}{|X|}\bar K_0(X,\phi,\rho) \cdot\chi_{\cS}(X)\Big|_{\phi=0} \\   
\alpha^{(2)}_{\mu,\nu}(X,\rho=0)=\frac{1}{2|X|}\bar K_2(X,\phi,\rho;x_{\mu},x_{\nu}) \cdot\chi_{\cS}(X)\Big|_{\phi=0\atop \rho=0}\\
\alpha^{(2)}_{\mu,\nu\rho}(X,\rho=0)=\frac{1}{2|X|}\bar K_2(X,\phi,\rho;x_{\mu},x_{\nu}x_{\rho}) \cdot\chi_{\cS}(X)\Big|_{\phi=0\atop \rho=0} \\
\end{cases}
\ee
Here $\chi_{\cS}$ is the characteristic function of small sets.
Later we will see this extraction is sufficient. 

Under the specified extraction, the first step of $\rho-$dependent extraction will be:
\be
\cE xp(\Box e^{-V}+K(\rho))(\La)=\cE xp(\Box e^{-V'(F(\rho))}+\cE(K(\rho),F(\rho)))(\La)
\ee
However, we expect the new renormalization group transformation goes as:
\be \begin{split}
\label{ge01}
&\int\cE xp(\Box e^{-V(\phi_L+\zeta)}+K(\phi_L+\zeta,\rho))(L\La)d\mu_C(\zeta)\\
=&e^{\de\cE(\rho)}\cE xp(\Box e^{-V''(\phi)}+K'(\phi,\rho))(\La,\phi)
\end{split}\ee
In other words, we want to keep $V_j\rightarrow V_{j+1}$ to be independent of $\rho$. Therefore we need to take further steps for the extraction as:
\be
\cE xp(\Box e^{-V}+K(\rho))(\La)=e^{\cE(\rho)}\cE xp(\Box e^{-V''}+\cE'(K(\rho),F(\rho)))(\La)
\ee
To accomplish this idea, we need some properties for these extraction terms. Note that
\be   
\begin{cases}
\sum_{X\supset\De}\alpha^{(0)}(X,0)=\de E \\  
\sum_{X\supset\De}\alpha^{(2)}_{\mu,\nu}(X,0)=-\frac{1}{2\beta}\de_{\mu\nu}\de\sigma \\
\sum_{X\supset\De}\alpha^{(2)}_{\mu,\nu\rho}(X,0)=0 \\
\end{cases}
\ee
where $\de E$ and $\de \sigma$ are constants. Then
\begin{align}
\cE(\rho)=\sum_{X\subset\La}F_0(X,\rho)
         =\sum_{X\subset\La} \alpha^{(0)}(X,\rho)
\end{align}
will be $\rho-$dependent. And by the 
fact  \ref{05262011}, 
\begin{align}
e^{-V''}(X)
&= \prod_{\De\subset X}e^{-V(\De)+\sum_{Y\supset \De}F_1(Y,\De)    } 
\end{align}
will be $\rho-$ independent.

Therefore the flow $V_j\rightarrow V_{j+1}$ will be $\rho-$absent and $\cE_j(\rho)\rightarrow\cE_{j+1}(\rho)$ and $K_j(\rho)\rightarrow K_{j+1}(\rho)$ will be $\rho-$ dependent.

Also the stability condition and necessary convergent condition hold similarly as in Theorem \ref{RGcontrol} if $\la_0$ is sufficiently small.

Similarly as the partition function, the extraction we made in our RG transformation is from $(\cS_1\cB_1 K(\rho))^{\natural}$. Then with the extraction specified above, by induction we have
\be
\dim \overline{K}_j(\rho) \geq 2\qquad \text{for any $j$}
\ee
This is not as good as 
\be
\dim \overline{K}_j(\rho=0) \geq 4\qquad \text{for any $j$}
\ee
but here it will suffice due to the pinning property.

\subsubsection{Bounds on $K_j(\rho)$}
The most important term turns out to be the analysis on $\cE_j$. But first of all, we need to make sure the existence of the RG flow, which requires the estimate on $K_j\rightarrow K_{j+1}$. 

A polymer activity $K(X,\phi,\rho)$ defined on $\La_{M-j}$ has pinning property at $\{a,b\}\subset\La_M$ if 
\[K(X,\phi,\rho)=K(X,\phi,\rho=0)\qquad\text{ for $X\cap\{L^{-j}a,L^{-j}b\}=\emptyset$}\] 
\begin{lem} \label{NKestimate}
For $0\leq j\leq M-1$,
\be
\norm{\cE_1\cF_1 K_{j}(\rho)}_{j+1} \leq \frac{3\de}{4}\norm{K_{j}(\rho)}_{j}
\ee
where 
$\de=\cO(1)\max(L^{-2},L^{2-\beta/4\pi})$.
\end{lem}
\begin{proof}
\mbox{}

Note that the new $\rho-$dependent extraction is made for the linear part of the neutral sectors from small sets, i.e, the  term $\cR_1(k_0 \cdot\chi_\cS)$. Therefore the estimate on the first two terms still holds, the rest is to control the growth of the $\rho-$dependent neutral sectors.

Note that 
\be \label{07192011d}
\cE_1\cF_1(k_0 \cdot\chi_\cS)=\cE_1(  \cF_1(k_0(\rho)\cdot\chi_{\cS}),F ) =\cF_1(k_0\cdot\chi_{\cS})-F(\cF_1(k_0\cdot\chi_{\cS}))  
\ee
where 
$F$ is the extracted part. Moreover, for simplicity let
\be
J(X,\phi,\rho)=\cF_1(k_0\cdot\chi_{\cS})(X,\phi,\rho)
\ee
and then write
\be
\begin{split}
J(X,\phi,\rho)=\underbrace{   J(X,\phi,\rho)-J(X,\phi,\rho=0)}_{\mbox{say as } J_1(X)} +\underbrace{J(X,\phi,\rho=0)}_{\mbox{say as }J_2(X)}
\end{split}
\ee
Now similarly as the vacuum case, we consider the term $\cF_1 k_0$ first.
\begin{enumerate}
	\item 
For term $J_1$:
By the pinning property 
\be
K_j(Y,\phi,\rho)=K_j(Y,\phi,\rho=0) \quad \mbox{if} \quad Y\cap\{L^{-j}a,L^{-j}b\}=\emptyset
\ee
we have
\be \begin{split}
&\cF_1 k_0(X,\phi,\rho)-\cF_1 k_0(X,\phi,\rho=0) \\
=&\sum_{\overline{ Y} ^L=LX \atop Y\cap\{L^{-j}a,L^{-j}b\}\neq \emptyset} \left(  k_0^{\#}(Y,\phi_L,\rho)- k_0^{\#}(Y,\phi_L,\rho=0)\right)
\end{split}\ee
since $k_0$ also have the pinning property.
Now note that
\be
\dim \left(  k_0(\rho) -k_0(\rho=0)  ) \right)\geq 2 \label{07152011}
\ee
Then by \ref{07122011c} and Lemma \ref{e0051},
\be
\begin{split}
&\norm{\cF_1 k_0(X,\rho)-\cF_1 k_0(X,\rho=0)}_{h,G(\kappa+\de\kappa)} \\
\leq&\sum_{\overline{ Y} ^L=LX \atop Y\cap\{L^{-j}a,L^{-j}b\}\neq\emptyset} \norm{  k_0^{\#}(Y,\rho)- k_0^{\#}(Y,\rho=0)}_{h,G_L(\kappa+\de\kappa)}\\
\leq&\sum_{\overline{ Y} ^L=LX \atop Y\cap\{L^{-j}a,L^{-j}b\}\neq\emptyset} \norm{ \cS( k_0^{}(L^{-1}Y,\rho)- k_0^{}(L^{-1}Y,\rho=0))}_{h,G_{L^{-1}}(\kappa)} 2^{\abs{Y}}\\
\leq&\sum_{\overline{ Y} ^L=LX \atop Y\cap\{L^{-j}a,L^{-j}b\}\neq\emptyset}L^{-2} \norm{  k_0^{}(Y,\rho)- k_0^{}(Y,\rho=0)}_{h,G(\kappa)} 2^{\abs{Y}}
\end{split}
\ee
Now by summing over all small polymers, the pinning property will suppress the growth in volume $L^2$ and  just gives a constant $\cO(1)$ term. Therefore we have
\be \label{07192011}
\begin{split}
\norm{J_1}_{j+1} 
\leq&\cO(1)L^{-2} \norm{  k_0^{}(Y,\rho)- k_0^{}(Y,\rho=0)}_{j}\\
\leq&\cO(1)L^{-2} \norm{  K(\rho)}_{j}\\
\end{split}
\ee

\item
For term $J_2$, by the extraction we made, 
we have
\be
\dim \left(  K^{\dagger}(\rho=0) -F(  K^{\dagger}(\rho=0)  ) \right)\geq 4
\ee
which is the same as the estimate \ref{06302011} in the vacuum case/ partition function. 
So for the term $J_2$, we have
\be
\norm{J_2}_{j+1} 
 \leq \cO(1)L^{-2}\norm{K_j(\rho)}_j  \label{07192011b}
\ee
\end{enumerate}
Then for the extraction part  $F(\cF_1(k_0\cdot\chi_{\cS}))$, it is bounded by
\be \label{07192011c}
\norm{F(\cF_1(k_0\cdot\chi_{\cS}))}_{j+1}\leq \cO(1)\norm{\cF_1(k_0\cdot\chi_{\cS})}_{j+1} \leq \cO(1) L^{-2} \norm{K_j(\rho)}_j
\ee
by combining term $J_1$ and $J_2$.

Substitute the estimate \ref{07192011}, \ref{07192011b} and \ref{07192011c} into the norm of \ref{07192011d}, we  have
\be
\norm{\cE_1\cF_1(k_0 \cdot\chi_\cS)}_{j+1}\leq  \cO(1)L^{-2}\norm{K_j(\rho)}_j  \leq  \frac{\de}{4}\norm{K_j(\rho)}_j
\ee

Combine with the similar  estimates on the  large set and charged terms, we have the claimed bound.
\end{proof}
This lemma provides the necessary shrinking conditions of the polymer activities to continue the RG transformations.
\subsubsection{Summary}
As Theorem \ref{RGcontrol} for the partition function, now we can summarize to the  result of the existence of R.G flow, which is also Theorem \ref{05182011}:
\begin{thm}  \label{07202011}
For sufficiently small $z$  and $0<\epsilon<1$,
\be \label{07162011}
Z'(\rho)=e^{\cE_j(\rho)}\int \cE xp(\Box e^{-V_j(\phi)} +K_j(\phi,\rho))(\La_{M-j})  d\mu_{\beta v_{M-j}}(\phi)
\ee
hold for any $0\leq j\leq M-1$, where $V_j$'s are $\rho-$independent, and the  polymer activities satisfy:
\be \label{07162011b}
\norm{K_{j}(\rho)}_{j}\leq \de^{j}\epsilon'
\ee
where $\epsilon'=\max(\epsilon,C_{\la})$ and $\de=\cO(1)\max(L^{-2},L^{2-\beta/4\pi})$.
And $K_j(\rho)$ has pinning property for any $j$.
\end{thm}
\begin{proof}
The proof is  similar to the proof of Theorem \ref{RGcontrol} and \ref{pproof}. The $K_j(\rho)$ in equation \ref{07162011} is well defined for any $j$. Now we prove the estimate \ref{07162011b} by induction. $j=0$ follows from Lemma \ref{07162011c}. Suppose true for $j$. Now as the vacuum case,
\be
K_{j+1}(\rho)=\cL_j K_j(\rho)+g_j(\sigma_j,K_j(\rho))
\ee
where $\cL_j K_j(\rho)$ is the linearization of $K_{j+1}(\rho)$ in $K_j(\rho)$. Then by Lemma \ref{NKestimate}
\be \label{07162011e}
\norm{\cL_j K_j(\rho)}_{j+1}=\norm{\cE_1\cF_1 K_j(\rho)}_{j+1}\leq \de \norm{K_j}_j\leq \de^{j+1}\epsilon'
\ee

Now for the higher order term  $g_j(\sigma_j,K_j)=\cR_{\geq 2}(K_j,\sigma_j)$, we use a Cauchy bound estimate as in Lemma  \ref{RGcontrol}. By \ref{residue00} and Lemma \ref{07182011}, we have
\be \begin{split}\norm{\cR_{\ge 2}(K_j)}_{j+1} \leq& \cO(1)D^{-2}\sup_{\abs{s}=D}\norm{\cR(sK_j,sF_j)}_{j+1}\\
\leq &\cO(1)D^{-1}L^2 (\norm{K_j}_j+\abs{\sigma_j/\de\kappa_j}) \label{07162011d}
\end{split}\ee
Recall that $j-$independent constant $D$ could be any number greater than 2( and satisfy the hypothesis of Lemma \ref{07182011}). Now we can choose $D=(\de^j\epsilon)^{-1/4}$ for example, then \ref{07162011d} becomes
\be \label{07182011d}\begin{split}\norm{\cR_{\ge 2}(K_j)}_{j+1} 
\leq &\cO(1)(\de^j\epsilon)^{1/4} L^2(\norm{K_j}_j+\abs{\sigma_j/\de\kappa_j})\\
\leq &\cO(1)\de^j\epsilon'\cdot(\de^j\epsilon)^{1/4} L^2\de\kappa_j^{-1}\\
\leq &\frac{1}{4}\de^j\epsilon'\cdot(\de^j)^{1/4}  2^{j}\kappa_0^{-1}\\
\leq &\frac{1}{4}\de^{j+1} \epsilon'
\end{split}\ee
Here we have used the condition
\be
\epsilon\leq\frac{1}{4^4}(\cO(1)L^{-2}\kappa_0\de)^4
\ee
and the  estimate
\be \begin{split}
(\de^j)^{1/4}  2^{j}\leq 1 
\Leftrightarrow   \de^j  \leq  2^{-4j}
\end{split}\ee
which is true for $L$ large.

Combine \ref{07182011d} with \ref{07162011e}, we have the estimate \ref{07162011b}.
\end{proof}
Our following correlation analysis only occurs  in the linear part of $K_{j+1}$ in $K_j$.
Also  we would like to absorb the constant factor $\cO(1)$ into the exponent as
\be
\begin{split}
\de\norm{K_j(\rho)}_j=\cO(1)\de_0\norm{K_j(\rho)}_j=\cO(1)\de_0^{\epsilon}\cdot\de^{1-\epsilon_0}\norm{K_j(\rho)}_j \leq \de_0^{1-\epsilon_0}\norm{K_j(\rho)}_j
\end{split}
\ee
for $L$ sufficiently large, where $\delta_0=\max(L^{-2},L^{2-\beta/4\pi})$.

\subsection{Main Result of the dipole$-$dipole decay}
\subsubsection{ Estimate}
Let
 \begin{align}
T_j(\rho)&=\cE_j(\rho)-\cE_{j-1}(\rho)-(\cE_j(\rho=0)-\cE_{j-1}(\rho=0))              \label{tdef}
\end{align}
Then it is easy to see
\be
\frac{\partial^2}{\partial \la_1\partial\la_2} \left[\cE_j(\rho)-\cE_{j-1}(\rho)\right] \Big|_{\la_1=\la_2=0}=\frac{\partial^2}{\partial \la_1\partial\la_2} \left[T_j(\rho)\right] \Big|_{\la_1=\la_2=0}
\ee
And from now on, throughout the rest paper, $\epsilon_0$ means a small quantity $0<\epsilon_0<1$. 

\begin{lem} First Estimate
\label{firstestimate}
\be 
\frac{\partial^2}{\partial\la_1\partial\la_2}T_j(\rho) \Big|_{\la_1=\la_2=0} =
0     \textrm{$\qquad$ if $0\leq j< I$}
\ee
where  $I$ is certain integer satisfying $L^I\leq\abs{a-b}< L^{I+1}$ and $0<\epsilon_0<1$.
\end{lem}
\begin{proof}
According to the arguments in section \ref{VEstimates}, the new $\rho-$dependent energy terms $\de\cE_j$ will have the form:
\be
\cE_j(\rho)-\cE_{j-1}(\rho)=\sum_{X\subset \La_{M-j}}\alpha_j^{(0)}(X,\rho)
\ee
where $\alpha_j^{(0)}$ is defined in \ref{06272011} with $K$ replaced by $(\cS_1\cB_1 K(\rho))^{\natural}$,
\be 
\begin{cases}
V_{F_0}(\De)&=\sum\limits_{X\supset\De} F_0(X,\De) = \de E(\rho)+\cO(1)L^2\sigma_j\De C(0)\\
V_{F_1}(\De)&=\sum\limits_{X\supset\De} F_1(X,\De)=-\frac{\de\sigma}{2\beta}\int_{x\in\De} (\partial\phi)^2   
\end{cases} 
\ee
   which is also similar to \ref{07142011d}.

 Then the $L^2\sigma$ terms canceled in $T_j(\rho)$, which makes $T_j(\rho)$ to be a function of $\de E(\rho)-\de E(\rho=0)$ only.

Now the  term $\cE_j(\rho)-\cE_{j-1}(\rho)$ is no longer a constant as the partition function case because of the external charges $\rho$.
And note that the extractions are only made from small sets $X$.  

We denote $d_j(a,b)=\abs{a-b}/L^j$ as the scaling distance between the two dipoles. Then before $I$ R.G. steps, $d_j(a,b)\geq 2^2$, the two dipoles cannot fall into a single small set. Hence
 \be\frac{\partial^2}{\partial\la_1\partial\la_2}[\cE_j(\rho)-\cE_{j-1}(\rho)] \Big|_{\la_1=\la_2=0} = 0 \qquad \mbox{for}\quad j< I. \label{07142011b}\ee
For $j\geq I$, we get a non-zero contribution. Note that there is no contribution from the $\sigma$ terms in this case.
\end{proof}

Here note that if $\de_0=\max\{L^{-2},L^{2-\beta/4\pi}\}$ and  $L^I\leq\abs{a-b}< L^{I+1}$, then
	 \begin{align}\de_0^I =\max\{L^{-2I},L^{(2-\beta/4\pi)I}\}= \max\{\abs{a-b}^{-2},\abs{a-b}^{2-\beta/4\pi} \}
	 \end{align}
which implies the distance factor $\abs{a-b}$.

Furthermore, we can get the estimate on  the power-law decay of the correlations  by improving bound on $\cE_j$ with a concrete pinning argument.

Since the energy terms are from the extraction, it is sufficient to derive the bound on
\be\de K_j(\rho)=K_j(\rho)-K_j(\rho=0)\quad\text{for all $j\geq I$}\ee 
As the flow $K_j\rightarrow K_{j+1}$, we need a better estimate for the flow $\de K_j\rightarrow \de K_{j+1}$.

Denote
\begin{align}
\de \cR (K(\rho))&=\cR (K(\rho))-\cR (K(0))\\
\de(\cE_1\cF_1)(K(\rho))&=\cE_1\cF_1(K(\rho))-\cE_1\cF_1(K(0))
\end{align}
and so on.

Since $K(\rho)$'s are still $2\pi-$ translation invariant with respect to $\phi$, we can make the Fourier expansion for $\de K(\rho)$.
Let
\be 
\de k_q(X,\phi,\rho)  = \frac{1}{2\pi} \int_{-\pi} ^{\pi} e^{-iq\Phi}\de K(X, \Phi + \phi,\rho) d \Phi 
\ee
 for each $q$.
Then
\be \de K(X, \phi)= \de k_0(X,\f)+ \sum_{q\neq 0}  \de k_q(X,
\phi)\ee
The terms with $q\neq 0$ are called the  charged terms and the $q=0$
term is called the  neutral term. 
\begin{lem} For $z$ sufficiently small, $\beta>8\pi$,
\be\norm{\de K_{j}(\rho)}_{j}\leq  L^{-2(1-\epsilon_0)j}\epsilon'\ee
\end{lem}
\begin{proof}
We use the similar induction argument as in the proof of Theorem \ref{07202011}. Suppose true for $j$. 

\begin{enumerate}
\item Estimates on the higher order terms:  

Note that
\be
\norm{\de\cR_{\geq 2}( K_j(\rho))}_{j+1} =\norm{\cR_{\geq 2}( K_j(\rho))-\cR_{\geq 2}( K_j(0))}_{j+1} 
\ee
For the higher order term $\cR_{\geq 2}( K_j(\rho))$, we still have the similar Cauchy bound \ref{07182011d}. Then by choosing $D=2(L^{-2j}\epsilon)^{-1/4}$, \ref{07182011d} becomes
\be \norm{\cR_{\geq 2}( K_j(\rho))}_{j+1} \leq\frac{1}{8} (\cO(1)L^{-2})^{j+1} \epsilon'\leq \frac{1}{8} L^{-2(1-\epsilon_0)(j+1)}  \epsilon' \ee
Similarly for the term $\cR_{\geq 2}( K_j(0))$. Therefore for $\de\cR_{\geq 2}( K_j(\rho))$, we have
\be \norm{\de\cR_{\geq 2}( K_j(\rho))}_{j+1}\leq \frac{1}{4} L^{-2(1-\epsilon_0)(j+1)}  \epsilon' \ee

	\item Estimates on the charged terms:
	
Note that there is no extraction made from charged terms, and  the pinning property is preserved under the renormalization group transformation.
	\be
	\begin{split}
	\de(\cE_1\cF_1)\left(   k_q(\rho)\cdot\chi_\cS\right)
	&= \cE_1\cF_1(k_q(\rho)\cdot\chi_\cS)-\cE_1\cF_1(k_q(0)\cdot\chi_\cS)\\
	&= \cF_1(k_q(\rho)\cdot\chi_\cS)-\cF_1(k_q(0)\cdot\chi_\cS)\\
	\end{split}
	\ee
	And for small set $X$,
	\be
	\begin{split}
	&\cF_1(k_q(\rho)\cdot\chi_\cS)-\cF_1(k_q(0)\cdot\chi_\cS)\\
	=&\sum_{\bar{Y}^L=LX}k_q^{\#}(Y,\phi_L,\rho)-\sum_{\bar{Y}^L=LX}k_q^{\#}(Y,\phi_L,0)\\
	=&\sum_{\bar{Y}^L=LX} \left(k_q^{\#}(Y,\phi_L,\rho)-k_q^{\#}(Y,\phi_L,0)\right)\\
	=&\sum_{\bar{Y}^L=LX \atop Y\cap \{L^{-j}a,L^{-j}b\}\neq\emptyset} \left(k_q^{\#}(Y,\phi_L,\rho)-k_q^{\#}(Y,\phi_L,0)\right)\\
		=&\sum_{\bar{Y}^L=LX \atop Y\cap \{L^{-j}a,L^{-j}b\}\neq\emptyset} \de k_q^{\#}(Y,\phi_L,\rho)
	\end{split} \ee
  The last step is by the pinning property of $K$ (so as $k_q$). Therefore instead of summing over all $L^2$ terms of polymers which span $LX$, the last summation is just over $\cO(1)$ terms. So
  \be
	\begin{split}
	\norm{\de(\cE_1\cF_1)\left(   k_q(\rho)\cdot\chi_\cS\right)}_{j+1} 
	&\leq \cO(1) \norm{ \left(  \de k_q(\rho)\cdot\chi_\cS\right)}_j\\
	&\leq \cO(1)e^{2N_C\abs{q}}\norm{\de  k_q(\rho)}_{j} 
	\end{split} \ee
	Therefore for $\beta>8\pi$
	\be
\begin{split}
\norm{\de(\cE_1\cF_1)(\sum_{q\ne 0}   k_q\cdot\chi_\cS)}_{j+1}
= &\norm{\de \cF_1\left(\sum_{q\ne 0}   k_q\cdot\chi_\cS\right)}_{j+1}      \\
\leq &\cO(1)\sum_{q\neq 0}e^{2N_C\abs{q}}\norm{\de k_q}_j\\
\leq &\cO(1) \sum_{q\neq 0}\left(e^{ -|q|(\beta C(0)-2N_{\beta C})  +\beta C(0)/2}\right)\norm{\de K_j}_j\\
\leq  &  \cO(1) L^{-\beta/4\pi}\norm{\de K_j}_j
\end{split} \ee
The last step is by the property of the covariance $C(0)$.
Therefore for the charged sectors, by induction we have
\be \begin{split}
	\norm{\de(\cE_1\cF_1)\left(\sum_{q\ne 0}   k_q\cdot\chi_\cS\right)}_{j+1} & \leq    \cO(1) L^{-\beta/4\pi}\norm{\de K_j}_j  \\
	&\leq    \cO(1) L^{-2}\norm{\de K_j}_j \\
	& \leq  \frac{1}{4} L^{-2(1-\epsilon_0)} \norm{\de K_j}_j \\
	&\leq\frac{1}{4} L^{-2(1-\epsilon_0)(j+1)}  \epsilon'
\end{split}\ee
if $\beta>8\pi$.

	\item Estimates on the neutral terms:
	\be
	\norm{\de(\cE_1\cF_1)( k_0 \cdot\chi_\cS)}_{j+1} \leq   \frac{1}{4} L^{-2(1-\epsilon_0)(j+1)} \epsilon'
	\ee
This bound follows from the estimate on the neutral sectors $K_j(\rho)$ similarly.	
	\item Estimates on the large set terms:
	\be
	\norm{\de(\cE_1\cF_1)( K_j \cdot\chi_{\overline{\cS}})}_{j+1} \leq   \frac{1}{4} L^{-2(1-\epsilon_0)(j+1)} \epsilon'
	\ee
	This  bound follows from \ref{07172011b} similarly.
\end{enumerate}

Therefore by combining the charged, neutral, large set and large field terms, we have
\be
	\begin{split}
	 \norm{\de K_{j+1}(\rho)}_{j+1}
	 \leq    L^{-2(1-\epsilon_0)(j+1)}\epsilon'
	\end{split}
	\ee
	for $\beta>8\pi$.
 \end{proof}

\begin{lem} Second estimate \label{07212011}
	\be
\abs{	 \frac{\partial^2}{\partial \la_1\partial\la_2} \left[\cE_M(\rho)\right] \Big|_{\la_1=\la_2=0}}=\cO(1)\epsilon'\abs{a-b}^{-2(1-\epsilon_0)}
	\ee
	\end{lem}
\begin{proof}
Note that $K(\cdot,\rho)$ is analytic with respect to $\la_i$ in the  small ball with radius $\la_0$, therefore we can apply the Cauchy Bound to get:
\begin{align}
	\abs{\frac{\partial^2}{\partial \la_1\partial\la_2} \left[\cE_M(\rho)\right] \Big|_{\la_1=\la_2=0}}&=
\abs{	\frac{\partial^2}{\partial \la_1\partial\la_2} \left[\sum_{j=1}^{M} T_j(\rho)\right] \Big|_{\la_1=\la_2=0}}\\
	&\leq \cO(1)\cdot\sum_{j=1}^{M}\sup_{\abs{\la_i}=\la_0} \norm{\de K_j(\rho)}\\
	&=\cO(1)\epsilon' \sum_{j=I}^M L^{-2(1-\epsilon_0)j}\\
	&=\cO(1)\epsilon'\abs{a-b}^{-2(1-\epsilon_0)}
	\end{align}
	\end{proof}

\subsubsection{Conclusion}

Now we summarize into our main result:

\begin{thm}  
Let $\beta > 8 \pi$,  
let  $L$ be sufficiently large, 
and  let $\abs{z}$ be  sufficiently small.
Then
\be
\abs{\left<\partial\phi(a)\partial\phi(b)\right>^T}=\cO(1)\epsilon'\abs{a-b}^{-2(1-\epsilon_0)}
\ee
where $\cO(1)$ and $\epsilon'$ are constants independent of $M$.
\end{thm}
\begin{proof}
Note that the equation (\ref{ge01}) holds for all $1\leq j\leq M-1$.  Now to complete the proof, we want to absorb all terms in last step into exponential factor.
We start with result \ref{07162011} in $M-1$ step:
\be \begin{split}
&e^{\cE_{M-1}(\rho)} 
 \int  \cE xp \left(\Box e^{-V_{M-1}(\phi)}+K_{M-1}(\phi,\rho)\right) \left(\La_{1}\right) d\mu_{\beta C} (\phi)\\
 =& e^{\cE_{M}(\rho)}  \cE xp \left(\Box +K_{M}(\rho)\right) \left(\La_{0}\right)
\end{split} \ee
since $\La_1$ coincides with the range L of covariance $C$ exactly. The above equation is a special case of the RG step \ref{ge01} with $\phi=\phi_L=0$ and  $e^{-V_M(\phi=0)}=1$. 

 Now with $\La_0=\bR^2/\bZ^2$, $\Box=\De_0$ is the only block in the final step, therefore
\be
\cE xp(\Box +K_M(\rho))(\La_{0})=1+K_M(\rho,\La_{0})
\ee
which implies
\be
\log Z(\rho)= \tilde{\cE}_M(\La_0,\rho)-\tilde{\cE_M}(\La_0,\rho=0) 
\ee
with 
\be\tilde{\cE}_M(\La_0,\rho)=\cE_{M}(\La_0,\rho)+\log(1+K_M(\La_0,\rho))\ee
Note that 
\be\begin{split}\abs{\log(1+K_M(\La_0,\rho))-\log(1+K_M(\La_0,0))}&\leq \cO(1)\norm{\de K_M(\rho)}_M\\
&\leq\cO(1)\epsilon'  L^{-2(1-\epsilon_0)M}\end{split}\ee
Hence by \ref{tcdef}, 
\be
\abs{\left<\partial\phi(a)\partial\phi(b)\right>^T}\leq \abs{	 \frac{\partial^2}{\partial \la_1\partial\la_2} \left[\cE_M(\rho)\right] \Big|_{\la_1=\la_2=0}}+\cO(1)\epsilon'  L^{-2(1-\epsilon_0)M}
\ee
then the result follows by Lemma \ref{07212011}.
\end{proof}
\begin{remark}
Note that for $\beta>8\pi$, the dipole$-$dipole decay behaves like the second order power decay. This supports the picture that over the extremely low temperature($\beta$ large), the two dimensional Coulomb gas behaves like dilute dipole gas. And the bound we have shown here is the upper bound. 
\end{remark}

\section{Appendix}
\subsection{Regulator Estimate} \label{regulatorestimatesection}
Now we prove the estimate \ref{regulatorestimate}. Consider for $0\leq t\leq 1$, the family of large field regulators
\be
G_t(X,\phi)=2^{t\abs{X}}G_L(X,\phi)^tG(X,\phi)^{1-t}
\ee
We want to show that for $0\leq t\leq 1$, 
\be
\int G_0(X,\phi+\zeta)d\mu_{tC}(\zeta)\leq G_t(X,\phi)
\ee
Then the estimate \ref{regulatorestimate} follows as the special case $t=1$.

Denote \begin{align}(\mu_C *G)(X,\phi)&=\int G(X,\phi+\zeta)d\mu_C(\zeta)\\
U(t,X)&=\log G_t(X,\phi)\\
C\left(\frac{\partial U}{\partial \phi},\frac{\partial U}{\partial \phi}\right)&=\int C(x,y)\frac{\partial U}{\partial \phi(x)}\frac{\partial U}{\partial \phi(y)}dxdy\\
\Delta_C U&=\frac{1}{2}\int C(x,y)\frac{\partial^2 U}{\partial \phi(x) \partial \phi(y)}dxdy
\end{align}
where $G_2(\phi;\zeta,\zeta)$ denotes the second order functional derivative in $\phi$.
Then note that 
\begin{align}
\label{05142011}&\De_CU+\frac{1}{2}C\left(\frac{\partial U}{\partial \phi},\frac{\partial U}{\partial \phi}\right)\leq \frac{\partial U}{\partial t}\\
\Rightarrow & \frac{\partial G_s}{\partial s}-\De_C G_s\geq 0 \quad\mbox{for $0\leq s\leq t$}\\
\Rightarrow & \mu_{(t-s)C}*\left(\frac{\partial G_s}{\partial s}-\De_C G_s\right)\geq 0 \quad\mbox{for $0\leq s\leq t$}\\
\Rightarrow & \frac{\partial}{\partial s}\mu_{(t-s)C}*G_s\geq 0\quad\mbox{for $0\leq s\leq t$}\\
\Rightarrow & \mu_{(t-s)C}*G_s\leq G_t\quad\mbox{for $0\leq s\leq t$}\\
\Rightarrow & \mu_{tC}*G_0\leq G_t
\end{align}
Therefore to show the estimate \ref{regulatorestimate}, it suffices to verify \ref{05142011}. Now substitute the definition of $U(t,X)$
\be
\begin{split}
U(t,X)=t\abs{X}\log 2+&\kappa\sum_{\Delta}\sum_{1\leq\abs{\alpha}\leq s}(tL^{2\abs{\alpha}-2}+1-t)\int_{\mathring{\Delta}}(\partial^{\alpha}\phi)^2\\
+&\kappa c\sum_{\abs{\alpha}=1}(tL+1-t)\int_{\partial X}(\partial^{\alpha}\phi)^2
\end{split}
\ee
into \ref{05142011}. 
 We investigate the term $C\left(\frac{\partial U}{\partial \phi},\frac{\partial U}{\partial \phi}\right)$ first. For simplicity, we write $Y=\sum_{\De\subset X}\mathring{\De}\subset X$. Then by direct computation,
\be \label{05242011}
\begin{split}
&C\left(\frac{\partial U}{\partial \phi},\frac{\partial U}{\partial \phi}\right)\\=
&2\kappa^2(tL^{2\abs{\alpha}-2}+1-t)^2\sum_{\alpha,\beta}\int_{Y\times Y} \partial^{\alpha+\beta}C(x,y)\partial^{\alpha}\phi(x)\partial^{\beta}\phi(y)dxdy\\
&+2\kappa^2c(tL^{2\abs{\alpha}-2}+1-t)(tL+1-t)\sum_{\alpha,\beta}\int_{Y\times \partial X} \partial^{\alpha+\beta}C(x,y)\partial^{\alpha}\phi(x)\partial^{\beta}\phi(y)dxdy\\
&+2\kappa^2c^2(tL+1-t)^2\sum_{\alpha,\beta}\int_{\partial X\times\partial X} \partial^{\alpha+\beta}C(x,y)\partial^{\alpha}\phi(x)\partial^{\beta}\phi(y)dxdy
\end{split}
\ee
where $\sum_{\alpha,\beta}$ means $\sum_{1\leq \abs{\alpha},\abs{\beta}\leq s}$.
Note that in \ref{05242011}, the term  
 \be I=\int_{Y\times Y} \partial^{\alpha+\beta}C(x,y)\partial^{\alpha}\phi(x)\partial^{\beta}\phi(y)dxdy\ee 
with $\abs{\alpha}=\abs{\beta}=1$ is subtle since there is no corresponding term in $\frac{\partial U}{\partial t}$ with $\abs{\alpha}=1$ to dominate it. Instead we use integration by parts. First by smoothness condition of $\phi$, 
we can replace $Y$ by $X$ and compute
 \be \label{05232011}
 \begin{split}
 I
 =&\int_{X\times X} \partial^{\alpha+\beta}C(x,y)\partial^{\alpha}\phi(x)\partial^{\beta}\phi(y)dxdy\\
 =&\int_X\left( \int_{\partial X} \partial^{\beta} C(x,y)\partial^{\alpha}\phi(x)\partial^{\beta}\phi(y)dx- 
 \int_{ X} \partial^{\beta} C(x,y)\partial^{2\alpha}\phi(x)\partial^{\beta}\phi(y)dx \right) dy\\
 =&\int_{\partial X}\left( \int_{\partial X}  C(x,y)\partial^{\alpha}\phi(x)\partial^{\beta}\phi(y)dy-\int_{ X}  C(x,y)\partial^{\alpha}\phi(x)\partial^{2\beta}\phi(y)dy\right)dx\\
 &-\int_{ X}\left( \int_{\partial X}  C(x,y)\partial^{2\alpha}\phi(x)\partial^{\beta}\phi(y)dy-\int_{ X}  C(x,y)\partial^{2\alpha}\phi(x)\partial^{2\beta}\phi(y)dy\right)dx
 \end{split}
 \ee 
Now under some conditions on the covariance $C$, each term in \ref{05232011} is dominated by $\partial U/\partial t$. The boundary term $\partial X$ in \ref{05232011} is the reason we have included the boundary term in the large field regulator \ref{gee}.

Since all the other terms in \ref{05242011} with $\abs{\alpha}\geq 2$ is dominated by $\partial U/\partial t$,  by summing over all possible $\alpha$,  $C\left(\frac{\partial U}{\partial \phi},\frac{\partial U}{\partial \phi}\right)$ which  is order $\kappa^2$ is bounded by $\partial U/\partial t$ which is order $\kappa$.

Similarly the term
\be
\begin{split}
&\De_C U\\
&= \sum_{1\leq\abs{\alpha}\leq s} \kappa(tL^{2\abs{\alpha}-2}+1-t)\int_{Y}\partial^{2\alpha} C(x,x)dx + \kappa c(tL+1-t)\int_{\partial X}\partial^{2}C(x,x)\\
&\leq \sum_{1\leq\abs{\alpha}\leq s} \kappa(tL^{2\abs{\alpha}-2}+1-t) \abs{Y}\sup_{x\in Y}\abs{\partial^{2\alpha} C(x,x)} \\
&\qquad+ \kappa c(tL+1-t)\abs{\partial X}\sup_{x\in \partial X}\abs{\partial^{2}C(x,x)}\\
&\leq \abs{X}\log 2
\end{split}
\ee
under some conditions on the covariance $C$ and $\kappa$.
 
Therefore, the equivalent condition of the inequality \ref{05142011} turns out to be the following quantities are sufficiently small:
\begin{align}
&\kappa L^{2s-2} \sup_{1\leq \abs{\alpha},\abs{\beta}\leq s} \sup_{x\in X}\abs{\partial_x^{\alpha}\partial_x^{\beta}C(x,x)}\\
&\kappa c^{-1}\sup_{0\leq \abs{\alpha},\abs{\beta}\leq s} \sup_{x\in X}\int_X\abs{\partial_x^{\alpha}\partial_y^{\beta}C(x,y)}dy\\
&\kappa c^{-1}\sup_{0\leq \abs{\alpha},\abs{\beta}\leq s} \sup_{x\in X}\int_{\partial X}\abs{\partial_x^{\alpha}\partial_y^{\beta}C(x,y)}dy
\end{align}
with some conditions on $\kappa$ to be specified. Now recall the covariance defined in Section \ref{FRC}, without loss of generality we may assume $g(x)$ attains its maximum at $x=0$ according to the $C^{\infty}$ condition on $g$. Then 
\be
\begin{split}
\int_X C(x)dx&=\int_X\int_1^L\frac{1}{l}u(x/l)dldy=\int_1^L\int_{l^{-1}X}u(y)dydl\\
&\leq \cO(1)\int_1^L l^{-2}L^2dl=\cO(1)L^2
\end{split}
\ee and for $\abs{\alpha}\geq 1$,
\be
\partial^{\alpha}_xC(x-y)=\int_1^L \frac{1}{l^{1+\abs{\alpha}}} u^{\alpha}\left(\frac{x-y}{l}\right)dl\leq \cO(1)
\ee
Also by the finite range property, $\int_X$ is bounded by $L^2$ and $\int_{\partial X}$ is  bounded by $L$. Therefore, the condition on $\kappa$ should be $\kappa c^{-1}L^{2s-2}$ with $s\geq 2$ is sufficiently small.

Also note that in the present paper, $\kappa_0=h_{\infty}^{-2}\leq \cO(L^{-s})$, so $\kappa_0 c^{-1}L^{2s-2}$ is sufficiently small. Similarly for $\kappa_j$. This completes the proof of regulator estimate \ref{regulatorestimate}.

\end{document}